\documentclass[a4paper,twoside,notitlepage,makeidx]{article}
\usepackage{a4}
\usepackage[english]{babel}
\usepackage{color,graphicx}
\usepackage{subfigure}
\usepackage{tikz-cd}
\usepackage{amsmath,amssymb,amsthm,amsfonts}
\usepackage{hyperref}
\usepackage{mathrsfs}
\usepackage{verbatim}
\usepackage{tabularx}
\usepackage{wasysym}
\usepackage{soul}
\usepackage{makeidx}

\newcommand{\norm}[1]{\left\lVert #1\right\rVert}

\newcommand{\txt}[1]{\quad\textnormal{#1}\quad}

\newcommand{\h}[0]{\mathcal{H}}

\newcommand{\p}[0]{\mathcal{P}}

\newcommand{\A}[0]{\mathcal{A}}
\newcommand{\V}[0]{\mathcal{V}}
\newcommand{\R}[0]{\mathbb{R}}

\newcommand{\C}[0]{\mathbb{C}\;}
\newcommand{\N}[0]{\mathbb{N}}

\newcommand{\OO}[0]{\mathcal{O}}

\newcommand{\uli}[1]{\underline{#1}}
\newcommand{\oli}[1]{\overline{#1}}

\newcommand{\ra}[1]{\stackrel{#1}{\longrightarrow}}

\newcommand*{\longhookrightarrow}{\ensuremath{\lhook\joinrel\relbar\joinrel\rightarrow}}
\newtheorem{theorem}{Theorem}[section]
\newtheorem*{theorem*}{Theorem}
\newtheorem{lemma}{Lemma}
\newtheorem*{lemma*}{Lemma}

\newtheorem*{claim*}{Claim}

\newtheorem{corollary}[theorem]{Corollary}
\theoremstyle{definition}
\newtheorem{definition}{Definition}
\newtheorem*{definition*}{Definition}

\newtheorem{remark}[theorem]{Remark}

\makeindex

\begin{document}

\begin{titlepage}

\begin{center}\Huge{\textbf{Topos-Theoretic Approaches to Quantum Theory}}\par
\vskip0.3cm
\large Part III Essay \par
\vskip1cm
Matthijs V\'ak\'ar\par
\vskip0.3cm
\large St John's College\par
\vskip 0.3cm
\large University of Cambridge\par
\vskip 0.3cm
2 May, 2012.
\end{center}
\vskip 3.5 cm
\small{ I declare that this essay is work done as part of the Part III Examination. I have read and understood the Statement on Plagiarism for Part III and Graduate Courses issued by the Faculty of Mathematics, and have abided by it. This essay is the result of my own work, and except where explicitly stated otherwise, only includes material undertaken since the publication of the list of essay titles, and includes nothing which was performed in collaboration. No part of this essay has been submitted, or is concurrently being submitted, for any degree, diploma or similar qualification at any university or similar institution.\\
\\
Signed .....................................}
\vskip 1cm
\small{\noindent Dr. A. Kuyperlaan 31, 3818JB Amersfoort, The Netherlands}
\newpage
\thispagestyle{empty}
\begin{center}\Huge{
\textbf{Topos-Theoretic Approaches to Quantum Theory}}\par
\end{center}
\newpage
\thispagestyle{empty}
\noindent \textbf{To the Chairman of Examiners for Part III Mathematics.}\\
\\
Dear Sir,
I enclose the Part III essay of .......................................\\
\\
Signed ....................................... (Director of Studies)
\clearpage
\begin{abstract}
Starting from a naive investigation into the nature of experiments on a physical system, one can argue that states of the system should pair non-degenerately with physical observables. This duality is closely related to that between space and quantity, or between geometry and algebra. In particular, it is grounded in the mathematical framework of both classical and quantum mechanics in the form of a pair of duality theorems by Gelfand and Naimark. These theorems allow us to construct a classical phase space, a compact Hausdorff space, for an algebra of classical observables. In the case of quantum mechanics, this construction breaks down due to the non-commutative nature of the algebra of quantum observables. However, we can construct a Hilbert space as the geometry underlying quantum mechanics.

Although this Hilbert space approach to quantum mechanics has proven to be very effective, it does have its drawbacks. In particular, they arise when one associates propositions with observables and investigates what kind of logical structure they form. One way of doing this is by realising the propositions as certain subsets of the phase space. In classical mechanics this procedure indeed gives one the structure one would expect: a Boolean algebra. However, although the case of quantum mechanics yields a well-behaved mathematical structure, an orthocomplemented lattice, the physical interpretation of this logic is rather subtle, due to its crude notion of truth.

Recent work by Isham, Butterfield, Doering, Landsman, Spitters, Heunen et al., attempting to address these problems, has led to an alternative method for dealing with non-commutative algebras of observables and, with it, to an alternative framework for quantum kinematics. Moreover, this approach stays much closer to our intuition from classical physics; in some sense, the motto is: Quantum kinematics is exactly like classical kinematics, that is, not in Set, but internal to some other topos!

This review paper gives an introduction to the subject.
\end{abstract}
\end{titlepage}

\clearpage
\tableofcontents
\clearpage
\section{Introduction}
A theory of physics should be grounded in experiments: we formulate a \emph{hypothesis}, we perform measurements on a prepared \emph{system} by carrying out a certain \emph{experiment} on it, and we compare the two to see if our experiment falsifies our hypothesis. Somehow, these three concepts should correspond (not necessarily bijectively) to the mathematical concepts of \emph{proposition}, \emph{state} and \emph{observable}, respectively. A theory of physics should specify which systems, experiments and hypotheses we deal with.

Starting from a naive investigation into the nature of experiments on a physical system, one can argue that, regardless of what the theory is, observables should pair non-degenerately with states. One can argue that the observables naturally embed in a C*-algebra on which the states act as linear functionals. The crucial difference between classical and quantum mechanics is found in the Heisenberg uncertainty relation, which entails that the quantum algebra should be a non-commutative one.

Theorems by Gelfand and Naimark allow one to construct a geometry, a so-called phase space, underlying the states: a compact Hausdorff space in the case of classical mechanics and a Hilbert space for quantum mechanics, given by the GNS-construction. For this reason, it might be useful to think of classical mechanics taking place internally in some category of topological spaces, whereas some category of inner product spaces is the domain of quantum mechanics. I mean this in the sense that the categorical structures at hand in both categories (e.g. objects, morphisms, (co)limits, monoidally closed structure) have similar physical interpretations.

This is a very reasonable but traditional approach to explaining the similarities between classical and quantum mechanics. However, it has its drawbacks. In particular, one can associate \emph{propositions} with observables and investigate what kind of logical structure they form. One way of doing this is by realising the propositions as certain subsets of the phase space. In classical mechanics this procedure indeed gives one the structure one would expect: a Boolean algebra. However, although the case of quantum mechanics, which was first dealt with by Birkhoff and von Neumann \cite{birlog}, yields a well-behaved mathematical structure, an orthocomplemented lattice, which has since been intensively studied, the physical interpretation of this logic is rather subtle, due to its crude notion of truth.

This, more than anything else\footnote{Of course, there are also other issues with the treatment of quantum mechanics in $\mathrm{Vect}$. For instance, the GNS-construction, unlike the commutative Gelfand-Naimark correspondence, does not extend to a functor. Even worse, it is non-canonical: the construction depends on the choice of designated states. Of course, one can take all states to obtain the universal GNS-representation. However, due to size issues, this is often not a practical Hilbert space to work with.}, raises the question of whether the category of vector spaces really is the right framework for a theory of quantum mechanics, at least if one wants to do logic. Recent work by Isham, Butterfield, Doering, Landsman, Heunen, Spitters et al. (e.g. \cite{buttop,doetop1,doetop2,doetop3,hlstop}) suggests that a fundamentally different approach might lead to new insights. In their work, the quantum phase space is not a Hilbert space, but a compact Hausdorff space\footnote{Rather, a completely regular locale.} internal to a topos different from $\mathrm{Set}$, namely the topos of (co)presheaves over the poset of commutative subalgebras of the algebra of observables\footnote{As we will explain, this is strongly inspired by Bohr's doctrine of classical contexts.}. One then constructs a quantum phase space as a topological space internal to this topos by a version of the commutative Gelfand-Naimark correspondence. In this way, the GNS-construction and, with it, the framework of vector spaces are avoided. This leads to an approach to quantum logic that is, as I will argue, more reasonable from an operational point of view. Moreover, apart from stressing a very strong parallel with classical mechanics, it also gives insight into the origin of fundamental differences between the classical and quantum worlds.

However, there are currently two competing approaches that realise this idea in slightly different ways. The first approach, which I shall call `the contravariant approach', was developed by Isham and Butterfield and later by Doering and Isham, and uses the topos of presheaves on the poset of subalgebras. The second approach, `the covariant approach', is due to Heunen, Landsman and Spitters and uses the topos of copresheaves on this poset. The two approaches were recently extensively compared by Wolters \cite{wolcom} and appear to be compatible in a number of ways.

This review paper attempts to give a brief outline of the subject, culminating in a comparison of the contravariant and covariant approaches. Rather than exploring new territory, the aim of this text is to give an overview, leaving out technicalities where possible and focusing on ideas rather than calculations. All omitted details can be found in the excellent papers in the bibliography.

\begin{quotation}
\emph{``Well, \emph{I} never heard it before,'' said the Mock Turtle; ``but it sounds uncommon nonsense.''\\ }\color{white}{lalalalalalalalalalalalalallalalalalala}\color{black} \emph{ (Lewis Carroll's `Alice's Adventures\\\color{white}{lalalalalalalalalalalalalallalalalalalalala}\color{black}  in Wonderland', Chapter X.)}
\end{quotation}
\textbf{Acknowledgement}\\
First of all, I would like to point out that very few of the results presented in this paper are
original. This text is mainly meant to serve as an introduction to the subject, reviewing established results.

Second, I would like to thank Benoit Dherin for his passionate lectures on
mathematical quantum mechanics, which first got me interested in the subject. His enthusiasm for the subject has been very contagious. My gratitude also goes to Chris Heunen and Dr. D.J.H. Garling for pointing me towards relevant literature and, of course, to my supervisor Professor P.T. Johnstone for making my year in Cambridge so much more pleasant by enabling me to study this material as part of my course.

\section{Naive Physics (Syntax)}\label{sec:syntax}
Let us forget about the standard formalism for theories of physics for a moment. In this informal section we shall take a naive operational point of view (in the sense of Bridgman) to motivate a minimal core of syntax\footnote{We will not make a strict distinction between syntax and semantics. In particular, we shall not work out the details of the syntactic side of the story. Instead we hope that our informal comments in this section will be enough to transfer some intuitions about the formalism used in theories of physics.} that should be given an interpretation in each theory of physics. We hope that this philosophical background will help the reader to appreciate the topos approaches to quantum physics.\\
\\
Physical theories originate in experimental practice. We perform experiments on prepared systems to obtain outcomes\footnote{The sort $\{\mathrm{outcomes}\}$ will be interpreted as some object of real numbers in a topos.}: we have sorts\footnote{In the sense of many-sorted logics.} $\{\mathrm{systems}\}$, $\{\mathrm{experiments}\}$ and $\p\{\mathrm{outcomes}\}$ and a function symbol
$$\{\mathrm{systems}\},\{\mathrm{experiments}\}\ra{\mathrm{Measurement}}\{\mathrm{outcomes}\}.$$
The same information contained in the outcomes of all measurements should be captured by the truth status of all hypotheses\footnote{In an interpretation in a topos this is obtained from the last line by the counit of the exponential adjunction.}: we can also formulate physics using a function symbol
$$\{\mathrm{systems}\},\{\mathrm{experiments}\},\p\{\mathrm{outcomes}\}\ra{\mathrm{Hypothesis}}\Omega.$$
Here, $\Omega$ is some abstract sort of truth values\footnote{In the case of classical mechanics, it is clear that it should just be the set $\{\bot,\top\}$. In our topos interpretation of quantum kinematics (i.e. everything except dynamics) it will be the subobject classifier of another topos.}, the sort $\p\{\mathrm{outcomes}\}$ represents the idea of an object of subobjects of $\{\mathrm{outcomes}\}$ and the term $\mathrm{Hypothesis}(X,Y,\Delta)$ should be read as the statement ``If we perform experiment $X$ on system $Y$, we will obtain a measured value in $\Delta\in\p\{\mathrm{outcomes}\}$.''

From a physical point of view it is clear that these pairings will generally be degenerate in some sense. Indeed, there might be two different experimental procedures (described in some alphabet) that yield the same measurement results on all physical systems or, equivalently, assign the same truth value to all associated hypotheses. In that case, we would like to think that we are actually measuring the same `physical quantity'. The equivalence classes under this equivalence relation should correspond one-to-one with what physicists call \emph{observables}\footnote{One can compare the distinction between experiments and observables with that between functions in intension and functions in extension.}. Let us write $\A$ for the sort of observables. Similarly, we would like to identify certain (descriptions of) systems when they agree on all hypotheses involving all experiments (or observables). We say that they are in the same \emph{state}. We write $\mathscr{S}$ for the sort of all states. These identifications should induce pairings, i.e., function symbols
$$\A,\mathscr{S}\ra{}\{\mathrm{outcomes}\}$$
and
$$\mathscr{S},\A,\p\{\mathrm{outcomes}\}\ra{\mathrm{Proposition}}\Omega.$$
The interpretation of these pairings should be non-degenerate in some sense. The important point is that, from purely practical, syntactic considerations, it follows that states and observables should be dual to each other.\\
\\
There is also a reasonable way to construct a logical structure out of these data. If we fix terms $a$ of sort $\A$ and $\Delta$ of sort $\p\{\mathrm{outcomes}\}$ in the function symbol $\mathrm{Proposition}$, we obtain a function symbol
$$\mathscr{S}\ra{[a\in\Delta]}\Omega.$$
This should be interpreted as a subobject of the object of states\footnote{In the formalism we will set up, we will also interpret these propositions as certain specific observables. The idea is that in practice (in the interpretation) we have a mono $\Omega\longhookrightarrow\{\mathrm{outcomes}\}$ (for instance, $\{0,1\}\longhookrightarrow \R$) so we can identify subobjects of $\mathscr{S}$ with certain arrows $\mathscr{S}\ra{}\{\mathrm{outcomes}\}$. If we choose our object of observables large enough, we may hope that these arrows are included. (Recall that under currying we have a map $\A\ra{}\{\mathrm{outcomes}\}^\mathscr{S}$.)}. We note that the propositions of this form should give rise to a kind of Lindenbaum-Tarski algebra related to the theory of physics, generated by the operations induced from $\mathrm{Sub}\mathscr{S}$.

\section{Canonical Formalism (Semantics)}
\label{sec:obsstapro}
\subsection{C*-algebras and their States in Physics}
Of course, we will not be able to derive many consequences from the informal foundations of the previous section. We have to narrow down the concepts of observable and state. We therefore turn to the theory of C*-algebras, where these ideas take concrete form. From now on, C*-algebras and $*$-homomorphisms are understood to be unital unless explicitly stated otherwise. The formalisms we will set up in this section for classical and quantum kinematics are among the canonical ones and are particularly popular among mathematical physicists. The reader might want to think of them as an attempt at constructing a semantics that represents the core syntax we presented in the previous section.

There are good physical (operational) arguments to support the idea that the set of observables can be embedded in (the set of self-adjoint elements of) a C*-algebra (over $\C$), which we shall also denote by $\A$. Here a polynomial $p(a)$ in $a\in\A$ is interpreted as first performing a measurement of $a$ and then applying the polynomial to the measurement outcome. The operational interpretations of $a+b$ and $a\cdot b$ for arbitrary observables $a$ and $b$ are more subtle. The measurement of the first observable might alter the state of the system we are trying to measure and therefore disturb the measurement of the second observable. This certainly does not agree with the interpretation of sums of observables in existing theories of physics. Identically prepared states do not really solve this problem either, according to the various no-cloning theorems known in physics (e.g. \cite{diecom}).

In this picture states $\omega$ are interpreted as giving the expectation value $\omega(a)$ for measurements of observables $a$ on the corresponding system. One can argue that these should form linear functionals on $\A$. (Linearity over polynomials in one observable, at least, should be clear.) To agree with the physical intuition of an expectation value, we want states to be positive functionals, i.e. they should yield positive results on positive elements of $\A$ (i.e. elements of the form $a^* a$). Then, the self-adjoint elements of $\A$ will be the candidates for observables. In this interpretation, the constant observable $1$ should have expectation value $1$ in every state. This means that $\omega(1)=1$, for all $\omega$. This leads to the following definition.
\begin{definition}[State on a C*-algebra] By a \emph{state on a C*-algebra} we mean a positive normalised linear functional. It is easily checked that these form a convex set. The extreme points of this set are called \emph{pure states}, while the term \emph{mixed state} is used to refer to a state that need not be pure\footnote{This terminology reflects the physics. Pure states represent the physical idea of a state about which we have complete knowledge, as far as our theory of physics allows. (E.g. in quantum mechanics, even pure states exhibit statistical behaviour.) Mixed states generally represent statistical (ignorance) ensembles of states.}.\end{definition}
For self-adjoint $a$, one can show that $\norm{a}=\sup_{\omega\in\mathscr{S}}|\omega(a)|$; in general, $\norm{a}^2=\sup_{\omega\in\mathscr{S}}\omega(a^*a)$. Moreover, it is easily verified that the states separate the elements of $\A$ and conversely, as desired. I hope that this brief digression gives the reader enough motivation to believe that a C*-algebra framework is not too restrictive for a theory of physics. An excellent motivation for the use of C*-algebras in physics (including all the details I omitted) can be found in chapter 1 of \cite{strmat}.

\subsection{Classical and Quantum C*-algebras}
As it turns out, the C*-algebra of observables of classical mechanics should be commutative, while we use a non-commutative C*-algebra to describe quantum mechanics. Strocchi \cite{strmat} gives a useful argument for this fact, starting from Heisenberg's uncertainty relation. It goes as follows. Recall that we gave a state $\omega$ on a C*-algebra $\A$ the following physical interpretation: if $a$ is a self-adjoint element of $\A$, then we interpret $\omega(a)$ as the expectation value of a measurement of $a$ on a system in state $\omega$. Note that the variance $\Delta_\omega(a)^2$ would therefore be $\omega(a^2)-\omega(a)^2$.

Now, Heisenberg's uncertainty relation states that
$$\Delta_\omega (q_j)\Delta_\omega(p_j)\geq \hbar /2,$$ where $q_j$ and $p_j$ denote, respectively, the position and momentum observables in the $j$-th direction. Since the bound given by the Heisenberg uncertainty relation is independent of the state, it is natural to assume that its origin can be found in the algebra of observables.

Let $a,b\in \A$ be such that $a^*=a$ and $b^*=b$. Put $\tilde a=a-\omega(a)1$ and $\tilde b=b-\omega(b)1$. Since $(\tilde a-i\lambda \tilde b)(\tilde a+i\lambda \tilde b)$ is a positive element for real $\lambda$, positivity of $\omega$ implies that
$$\Delta_\omega(a)^2+\lambda^2\Delta_\omega(b)^2+i\lambda\omega([a,b])\geq 0,$$
where $[a,b]:=ab-ba$. This tells us that the last term is real and therefore, by non-negativity of the quadratic form in $\lambda$, $4\Delta_\omega(a)^2\Delta_\omega(b)^2\geq |\omega(i[a,b])|^2$ or, put differently,
$$\Delta_\omega(a)\Delta_\omega(b)\geq |\omega([a,b])|/2.$$
It is natural to demand that the algebraic bound, which follows from the mathematical formalism, coincides with the experimentally determined bound given by the Heisenberg uncertainty relation: $\hbar=|\omega([q_j,p_j])|$, for all states $\omega$. This implies that
$$[q_j,p_k]=\pm i\hbar \delta_{jk}1,$$
where $\delta_{jk}$ denotes the Kronecker delta. Therefore, we conclude that the algebra of quantum observables should be a non-commutative one\footnote{This heuristic argument also shows one of the weaknesses of the C*-algebra formalism. Indeed, if one has two self-adjoint elements $q,p\in\A$ such that $[q,p]=\lambda\cdot 1$, where $\lambda$ is some constant, one easily derives that $\norm{q}\norm{p}\geq n|\lambda|/2$, for all $n\in \N$, i.e. this cannot happen. In practice this means that either $q$ or $p$ will be an unbounded operator, thereby falling outside our C*-algebra formalism. However, one can argue that because of the practical restrictions on our measurements and relativistic considerations, the observables we measure in practice are in fact bounded approximations of $q$ and $p$. Cf. the next footnote.}.\\
\\
Our algebra of observables is a very abstract entity. As is often the case in algebra, we would like to think of the elements of our algebra as functions of some kind. The typical example of a commutative C*-algebra one might come up with is perhaps $\C$, or, more generally, some algebra of bounded continuous functions from a fixed space to $\C$. In the non-commutative case one would perhaps think of an algebra of matrices with complex coefficients or, more generally, an algebra $\A\subset \mathcal{B}(\h)$ of bounded operators on some Hilbert space $\h$. As it turns out, these are, in fact, the only examples. This was proved by Gelfand, Naimark (and I. Segal).

\begin{theorem}[(Commutative) Gelfand-Naimark theorem,\cite{johsto}]\label{thm:comgelnai} Write $\mathrm{cCStar}$ for the category of unital commutative complex C*-algebras with unital $*$-homomorphisms and $\mathrm{CHaus}$ for that of compact Hausdorff topological spaces. Then the functors
$$\mathrm{cCStar}^{op}\mathop{\leftrightarrows}^{C(-,\C)}_{\mathrm{Max}}\mathrm{CHaus}$$
define an equivalence of categories, where $C(-,\C)\subset \mathrm{Set}(-,\C)$ sends $\Sigma$ to the algebra of continuous functions from $\Sigma$ to $\C$ and $\mathrm{Max}$ sends an algebra $A$ to its maximal spectrum, equipped with the Gelfand topology and a unital $*$-homomorphism $A\ra{f}A'$ to a continuous map $\mathrm{Max}(A')\ra{\mathrm{Max}(f)}\mathrm{Max}(A)$ that sends a maximal ideal to its inverse image under $f$.
\end{theorem}
Maximal ideals in a commutative unital C*-algebra correspond precisely to non-zero unital $*$-homomorphisms to $\C$, or pure states. (See theorem \ref{thm:maxsp}.) When viewed this way, the maximal spectrum is commonly referred to as the Gelfand spectrum of the C*-algebra.

We would like to think of this space $\Sigma$ as the phase space of classical mechanics\footnote{A physicist might object that the phase space in classical mechanics is usually a cotangent bundle of some manifold $Q$ and therefore not compact. The compactness of our space reflects the fact that we started out with a normed algebra of observables. This means that we restricted our observables to be bounded continuous functions on our phase space. This is a reasonable thing to do since we only expect our classical physics to be applicable in a relatively small region of space. (At large distances general relativistic effects become relevant.) Similarly, we do not expect classical physics to hold for particles with very high momentum. In that domain one would also have to use a theory of relativity. One way to think of the relation between our compact phase space $\Sigma$ and the usual space $T^*Q$ from physics is that $C_b(T^*Q,\C)\cong C(\Sigma,\C)$, or equivalently $\Sigma=\mathrm{Max}(C_b(T^*Q,\C))$, where $C_b(T^*Q,\C)$ is the C*-algebra of bounded continuous functions $T^*Q\ra{}\C$. Note that this is just the definition of the Stone-\v Cech compactification! Since $T^*Q$ is locally compact Hausdorff, it embeds in $\Sigma$ as an open subspace (by the unit of the Stone-\v Cech-adjunction). \cite{johsto}} and of $C(\Sigma,\C)$ as the algebra of observables.
In view of the duality between states and observables, we should also explain what the geometric analogue of a mixed state is. The answer is given by a celebrated theorem by Riesz and Markov. This brings us back to the familiar situation in (statistical) classical physics, where states are probability distributions on phase space.
\begin{theorem}[Riesz-Markov,\cite{strmat}] Let $\Sigma$ be a compact Hausdorff space. Then there is a one-to-one correspondence between states $\omega$ on $C(\Sigma,\C)$ and Radon probability measures\footnote{That is, finite regular Borel measures of total mass $1$.} $\mu_\omega$ on $\Sigma$, the correspondence being
$$\omega(f)=\int_\Sigma f d\mu_\omega. $$
\end{theorem}

The construction for non-commutative C*-algebras is not nearly as satisfactory, since it is non-functorial.
\begin{theorem}[GNS-construction,\cite{strmat}] Any C*-algebra $\A$ is *-isomorphic to an algebra of bounded operators $\A'\subset \mathcal{B}(\h)$ on some\footnote{It can be chosen to be separable if $\A$ is.} Hilbert space $\h$. There is a canonical\footnote{We take the sum of the GNS-representations corresponding to all states.} construction to realise this, called the (universal) GNS-representation.\end{theorem}
This construction is an argument commonly given in favour of the Hilbert space framework in which quantum mechanics is usually formulated. States, as physicists know them, fit in this framework as follows.
\begin{corollary}Let $\A$ be a C*-algebra with any representation $\A\ra{\pi}\mathcal{B}(\h)$ on a Hilbert space $\h$. Then every positive trace-class operator $\rho$ on $\h$ with trace $1$ defines a state $\omega_\rho$ of $\A$ by
$$\omega_\rho(a):=\mathrm{tr}(\rho\pi(a)).$$
Conversely, if we take $\pi$ to be the universal GNS-representation, every state in fact arises (non-uniquely) in such a way; in particular, each pure state is represented by the rank-one projection onto its cyclic GNS vector.\end{corollary}
Given a general representation $\A\ra{}\mathcal{B}(\h)$ of a C*-algebra $\A$, however, our notion of a state is strictly more general than that of states obtained from tracing against positive trace-class operators. The difference lies mostly in the fact that our states are only finitely additive, while in physics normality, or complete additivity on projections, is conventional: whenever the ultraweak sum of mutually orthogonal projections $P_i\in\A$ exists, one asks that $\omega(\sum_iP_i)=\sum_i\omega(P_i)$.
Although our concept of a state might therefore be somewhat more general than the states that physicists use in practice, it does not differ much at a conceptual level. On the other hand, it is much easier to deal with mathematically. This shows the physical relevance of the abstract algebraic concept of a state on a C*-algebra that we shall use in the rest of this paper.

\subsection{Propositions}
\subsubsection{Classical Propositions}\label{sec:clprop}
To motivate the rest of this paper, which will mostly concern quantum logic, we give a brief account of the logic associated with classical mechanics.

Let us approach the matter naively. When we perform experiments on a classical mechanical system, we formulate hypotheses about it. It is not uncommon to manipulate these hypotheses by logical operations such as negations, disjunctions, conjunctions and implications, the last of which are very important if we want to test physical theories\footnote{As we shall see, the lack of a suitable notion of an implication is one of the major drawbacks of von Neumann's quantum logic.}. The idea is that we should be able to translate these logical operations in our metalanguage into operations in the mathematical framework of classical mechanics. In section \ref{sec:syntax}, we argued that it is reasonable to expect to embed propositions in the powerset of our set of states. In fact, we will embed them in the powerset of the set $\mathscr{S}$ of pure states. Naively, we would expect to find a logical structure in classical mechanics that reflects our classical intuition: a Boolean algebra. One might hope to realise this structure as a Boolean subalgebra of $\p\mathscr{S}$.\\
\\
To proceed, we take an operational point of view again. A hypothesis about a classical mechanical system should be verifiable, or at least falsifiable\footnote{On the level of logic, one can argue that this difference is not that important. Indeed, verification of a proposition would coincide with the falsification of its negation.}. It would therefore typically be a statement of the form ``If we measure the observable $a$ on our system, the outcome will be in $\Delta\subset \R$.'' This can equivalently be expressed in terms of geometry (Gelfand-Naimark, Riesz-Markov) as ``The state of our system has support in $a^{-1}(\Delta)\subset\Sigma$.''

Now, a mathematician would ask: ``How free are we in choosing $\Delta$?'' One reasonable answer is that $\Delta$ should be open in $\R$. This gives us $\OO(\Sigma)$ as the object representing our classical logic. However, in general, this is only a Heyting algebra: the law of the excluded middle fails. Heunen, Landsman and Spitters suggest in \cite{hlstop} that it \emph{almost} holds: we would never be able to establish its failure by performing experiments. Indeed, they reason, if $P$ is a proposition and $U\in\OO(\Sigma)$ is the corresponding open set, then $U\cup (\Sigma\setminus U)^{o}$ would be the open set representing $P\vee \neg P$ and this set is topologically big (dense) in $\Sigma$. However, since physical probabilities correspond to certain integrals on $\Sigma$ with respect to some Borel measure, I would say that being `big' in a measure-theoretic sense is a better criterion. Note that $U\cup (\Sigma\setminus U)^{o}$ might not have full measure, so this is not the most reasonable logic we could associate to classical mechanics. The law of the excluded middle could really fail in a physically detectable way in this logic.\\
\\
Therefore, if one wants to end up with a Boolean algebra, it might be better to allow $\Delta$ to be a general Borel-measurable subset of $\R$ and to replace $\A=C(\Sigma,\C)$ by the C*-algebra of essentially bounded functions $L^\infty((\Sigma,\mu),\C)$, with respect to some Borel measure $\mu$ on $\Sigma$. Then our logic is given by the Borel $\sigma$-algebra (modulo null sets)\footnote{Another obvious way to get a Boolean algebra, mentioned in some references such as \cite{hlstop}, would be to consider all subsets of $\Sigma$. However, this choice does not seem to reflect very well what can be established by experiments, according to the measure-theoretic interpretation of states given by the Riesz-Markov theorem.}.

Note that we have a canonical map $C(\Sigma,\C)\longhookrightarrow L^\infty((\Sigma,\mu),\C)$ that is injective if $\mu$ is non-degenerate \cite{takthe}. Moreover, the logic of our classical mechanics does not come from the geometry alone. It can equivalently be described in algebraic terms as the Boolean algebra $\Pi(\A)$ of self-adjoint idempotents\footnote{One can check that these form a complete Boolean algebra, where the meet is given by the ring multiplication. \cite{redqua} In this case these are of course precisely the characteristic functions of measurable subsets of $\Sigma$.} in $L^\infty((\Sigma,\mu),\C)$. These results can be stated more generally as the following non-trivial theorem. (In our particular case $\h=L^2((\Sigma,\mu),\C)$ and $L^\infty((\Sigma,\mu),\C)\subset \mathcal{B}(\h)$ (by multiplication) is our von Neumann algebra.)

\begin{definition}[von Neumann algebra] \label{def:vonneu} A C*-algebra is called a von Neumann algebra if it admits an embedding into $\mathcal{B}(\h)$ for some Hilbert space $\h$ such that it is equal to its double commutant. We will think of von Neumann algebras as a non-full subcategory\footnote{We will not think of a Hilbert space as being part of the data of a von Neumann algebra. The reader should note that in the literature, these von Neumann algebras without an embedding in some $\mathcal{B}(\h)$ are also known as W*-algebras.} $\mathrm{VNeu}$ of $\mathrm{CStar}$ where the $*$-homomorphisms are additionally required to be continuous in the ultraweak topologies\footnote{This is the initial topology with respect to the family of maps $\mathcal{B}(\h)\ra{\mathrm{tr}(\rho - )}\C$, for trace-class operators $\rho\in \mathcal{B}(\h)$. It turns out that this topology does not depend on the choice of the Hilbert space \cite{takthe}}.\end{definition}
In this context it might be easiest to think of von Neumann algebras as certain C*-algebras that have enough self-adjoint idempotents. In fact, these generate a von Neumann algebra in the sense that the double commutant of $\Pi(\A)$ is $\A$ \cite{takthe}. This abundance of projections\footnote{In the light of the universal GNS-representation, it is reasonable to call self-adjoint idempotents projections.} will make our lives easier in many ways\footnote{For example: \begin{enumerate}
\item By the spectral theory for von Neumann algebras, we can extend many results about projections to all normal elements, so in particular to all observables.
    \item The projections will represent the propositions about our physical system. The abundance of projections tells us that this logic contains a lot of information about the system.
        \item This logic will also be convenient from a calculational point of view: it is a complete lattice.\end{enumerate}}.

The step of passing from $C(\Sigma,\C)$ to $L^\infty((\Sigma,\mu),\C)$ is somewhat unsatisfactory, since it depends on the choice of a measure $\mu$. The motivation for this construction was that we wanted to add self-adjoint idempotents (specifically, characteristic functions of measurable sets) to the algebra of observables, to obtain an interesting logic from the algebra. As it turns out, there is a more canonical way of doing this.
\begin{theorem}[\cite{lurlec},\cite{takthe}]$\mathrm{VNeu}\longhookrightarrow\mathrm{CStar}$ is a (non-full) reflective subcategory. Its reflector is called the universal enveloping von Neumann algebra construction. It is given by applying the universal GNS-representation and taking the bicommutant, or, equivalently, by taking the second continuous dual space. \end{theorem}
The result for our particular case $\A=C(\Sigma,\C)$ indeed reflects our intuition.
\begin{theorem}Recall that the Radon measures on $\Sigma$ form a preorder under the relation of absolute continuity. $L^\infty((\Sigma,-),\C)$ defines a functor from the opposite of this preorder to the category of complex Banach spaces. Then $C(\Sigma,\C)^*{}^*\cong \varprojlim L^\infty((\Sigma,-),\C)$. This is a von Neumann algebra under the induced multiplication.\end{theorem}
\begin{proof}[Proof (sketch)] It is well known from the construction of Riesz-Markov that $C(\Sigma,\C)^*\cong \varinjlim L^1((\Sigma,-),\C)$, where $L^1((\Sigma,-),\C)$ is again a functor from the opposite of the preorder of Radon measures to the category of complex Banach spaces. Now, the continuous dual space functor sends colimits to limits (this is easy to see by considering it as a contravariant endofunctor of the category of topological vector spaces, of which that of Banach spaces is a full subcategory) and therefore $C(\Sigma,\C)^*{}^*\cong \varprojlim (L^1((\Sigma,-),\C))^*$. Now, according to \cite{takthe}, $(L^1((\Sigma,-),\C))^*\cong L^\infty((\Sigma,-),\C)$. Finally, as a consequence of Hahn-Banach, $C(\Sigma,\C)$ is w*-dense in its bidual space. Therefore, by continuity, the multiplication on $C(\Sigma,\C)^*{}^*$, which we know to exist by the previous theorem, has to coincide with the one induced from the multiplications on each of the $L^\infty((\Sigma,\mu),\C)$.\end{proof}
Note that this space contains the coherent families determined by bounded Borel measurable functions from $\Sigma$ to $\C$.\footnote{Of course, it contains a lot more. For instance, objects that do not have an interpretation as functions.} Therefore, the characteristic functions define an order embedding of the Borel $\sigma$-algebra into the self-adjoint idempotents of this von Neumann algebra\footnote{The self-adjoint idempotents of a von Neumann algebra are a complete lattice, while the Borel $\sigma$-algebra may only be countably complete. So this embedding is a kind of completion (with objects that do not have an interpretation as subsets of $\Sigma$).}. Effectively, we have added those characteristic functions to $C(\Sigma,\C)$ and embedded the result in a C*-algebra.
One might wonder how the abundance of projections in a von Neumann algebra translates into properties of the spectrum.
\begin{theorem}[Gelfand duality for commutative von Neumann algebras]\label{thm:hypsto} The equivalence of Theorem \ref{thm:comgelnai} restricts to an equivalence between the opposite of the category of commutative von Neumann algebras and the category of hyperstonean\footnote{i.e. extremally disconnected and admitting a perfect measure. What is important here is that it is a stronger condition than being Stone (see below).} spaces, where the maps are open continuous functions \cite{takthe}. \\
\\
\emph{To accentuate the relation with measure theory, we also mention the following classification.}\\
\\
If $\Sigma$ is a locally compact Hausdorff space and $\mu$ is a Radon measure on $\Sigma$, then $L^\infty((\Sigma,\mu),\C)$ is a commutative von Neumann algebra. Conversely, every commutative von Neumann algebra is of this form for a suitable measure space \cite{sakcst}.
\end{theorem}
From the fact that the Gelfand spectrum of a commutative von Neumann algebra is a Stone space, one finds a second way of constructing a Boolean algebra out of it. Recall Stone's representation theorem for Boolean algebras.
\begin{theorem}[Stone-Duality,\cite{johsto}] Write $\mathrm{Bool}$ for the category of Boolean algebras and homomorphisms and write $\mathrm{Stone}$ for the category of zero-dimensional\footnote{That is, compact Hausdorff spaces with a basis of clopen subsets.} compact Hausdorff spaces and continuous functions. Then the functors
$$\mathrm{Bool}^{op}\mathop{\leftrightarrows}^{\p_{cl}}_{\mathrm{Spec}}\mathrm{Stone}$$
define an equivalence of categories, where $\p_{cl}\subset\p$ sends $\Sigma$ to the Boolean algebra of clopen subsets and $\mathrm{Spec}$ sends a Boolean algebra $B$ to its prime spectrum, equipped with the Zariski topology and an algebra morphism $B\ra{f}B'$ to a continuous map $\mathrm{Spec}(B')\ra{\mathrm{Spec}(f)}\mathrm{Spec}(B)$ that sends an ideal to its inverse image under $f$.\end{theorem}
It turns out that this gives us $\Pi(\A)$ again. We formulate this in terms of spectra.
\begin{theorem}[\cite{bezsto}] The Gelfand spectrum of a commutative von Neumann algebra coincides with the Stone spectrum of the Boolean algebra of its self-adjoint idempotents: $\mathrm{Spec}\circ \Pi=\mathrm{Max}$ on commutative von Neumann algebras\end{theorem}
In this sense, (commutative) von Neumann algebras are suitable algebras of observables for classical mechanics, from a logical point of view. We can reconstruct the observables from their associated logic: $C(\mathrm{Spec}(\Pi(\A)),\C)\cong \A$.

\subsubsection{von Neumann's Quantum Propositions}\label{sec:bn}
The construction in the previous paragraph was designed so that it would yield a logic for classical mechanics that agrees with our classical (Boolean) intuition. Things get more interesting if we try to do something similar for quantum mechanics, as it is not a priori clear that our classical logic is suitable at all for dealing with statements about quantum mechanical systems. The first thorough account of such a quantum logic was given by Birkhoff and von Neumann in \cite{birlog}. The approach that is now standard is only a minor alteration of their original one\footnote{It is now generally accepted that quantum logic should be orthomodular only, not modular, i.e. distributivity fails in an even stronger sense.} \cite{enghan}.

Inspired by the discussion of classical logic, we might postulate that the poset $\Pi(\A)$ of self-adjoint idempotents of our C*-algebra of quantum observables $\A$ should model quantum logic, playing the role of a Lindenbaum-Tarski algebra. Mimicking the situation in classical mechanics, we define a partial order on $\Pi(\A)$ by $a\leq b$ iff $ab=a$. (It then follows that $ba=b^*a^*=(ab)^*=a^*=a$.)\footnote{Geometrically speaking, i.e. if $\A\subset \mathcal{B}(\h)$, this order is the inclusion of subspaces.} However, in general, $\A$ might not contain enough projections to give rise to an interesting logic, just like $C(\Sigma,\C)$ did not in the classical case. Therefore, we should embed $\A$ in some von Neumann algebra, the quantum analogue of $L^\infty(\Sigma,\C)$, if you will.

One way of doing this, if $\A$ arises as a subalgebra of $\mathcal{B}(\h)$ for some Hilbert space $\h$, is just to take its double commutant. A canonical method would be to use the universal GNS-representation for $\A$ to embed it in $\mathcal{B}(\h)$ for some Hilbert space $\h$, thus computing the universal enveloping von Neumann algebra. Then we have the following.
\begin{theorem}[\cite{redqua}] The lattice $\Pi(\A)$ of projections in a von Neumann algebra $\A$ is a complete orthomodular\footnote{The orthomodular law says that, if $A\leq B$, then $B=A\vee(A^\bot\wedge B)$, where we write $A^\bot$ for $1-A$ (in the algebraic sense). Note that this is not a Heyting implication.} lattice, with respect to its natural order. Moreover, it is distributive if and only if $\A$ is commutative. In particular, in the non-commutative case, the adjoint functor theorem tells us that it does not have an implication.\end{theorem}
\begin{proof}All properties are straightforward verifications, except completeness. This is an elementary consequence of the (non-trivial) bicommutant theorem, which I have therefore used to define what a von Neumann algebra is.\end{proof}

Note that a faithful representation of $\A$ on a Hilbert space $\h$ defines an order embedding $\Pi(\A)$ into the Grassmannian $\Pi(\h)$ of closed linear subspaces of $\h$. We see that this quantum logic resembles the situation in classical mechanics quite closely in terms of geometry: the propositions are given by the closed linear subspaces of $\h$ (rather than by measurable subsets). The meet is given by intersection and the join by closed linear span. Although an implication is lacking, Birkhoff and von Neumann did define a negation on $\Pi(\A)$, given by the orthocomplement (or, $p\mapsto 1-p$ in terms of algebra).\\
\\
As we anticipated in our informal syntactic discussion in section \ref{sec:syntax}, this logic should play the role of a Lindenbaum-Tarski algebra and its propositions should arise as (in some sense provable) equivalence classes of sequents $\top \vdash a\in\Delta$, for variables $a$ of sort $\A$ and $\Delta$ of sort $\p\R$. In practice this goes as follows.

Let $a\in \A\subset\mathcal{B}(\h)$, $a^*=a$ be a quantum observable and let $\Delta$ be a Borel-measurable subset of $\R$. Then we can form a proposition $[a\in \Delta]\in \Pi(\A)$ as the projection in the spectral decomposition of $a$ corresponding to $\Delta\cap\sigma(a)\subset \sigma(a)$, where $\sigma(a)$ denotes the spectrum of $a$. Conversely, every $p\in\Pi(\A)$ is of this form. Indeed, take $a=p$ and $\Delta=\{1\}$.

A first indication of the physical interpretation of this logic is found in the following useful but quite non-trivial analogy\footnote{One should immediately note, however, that a big difference is that quantum pure states do not correspond to the two-valued measures: even pure states exhibit non-deterministic behaviour.} with the Riesz-Markov representation of states from classical mechanics.
\begin{theorem}[Mackey-Gleason,\cite{bunmac,bunmac2,bunmac3} (original references),\cite{hamqua} (complete proof)]\label{thm:macgle} Let $\A$ be a von Neumann algebra. Every state $\omega$ on $\A$ restricts to a finitely additive measure on $\Pi(\A)$. That is, a map $\Pi(\A)\ra{\mu_\omega}[0,1]$ such that
\begin{enumerate}
\item $\mu_\omega(\top)=1$;
\item $\mu_\omega(x)+\mu_\omega(y)=\mu_\omega(x\vee y)$,
\txt{whenever $x\leq y^\bot$ (equivalently, $y\leq x^\bot$).}
\end{enumerate}
If $\A$ has no type $I_2$ summands, this assignment defines a bijection between states\footnote{In fact, it always defines a bijection between quasi-states and these finitely additive measures. (See definition \ref{def:quasta})} on $\A$ and finitely additive measures on $\Pi(\A)$.\end{theorem}
\noindent We see that the choice of a state associates a number to every quantum proposition. The Born interpretation (one of the cornerstones of quantum mechanics that gives the formalism its empirical content) precisely states that these numbers should be interpreted as (frequentist) probabilities\footnote{This is a rather reasonable assumption, since this is precisely our interpretation in the case of classical mechanics.} and that, for a state $\A\ra{\omega}\C$, the probability of measuring a value in $\Delta\subset\R$ for the observable $a$ is
$$\mathrm{Prob}_\omega(a\in\Delta)=\mu_\omega([a\in\Delta])=\omega([a\in\Delta]).$$

As it turns out, we identify a quantum proposition with the pure states for which it is true with certainty. The Born rule tells us that, for a pure state represented by a unit vector $|\psi\rangle\in\h$, this probability is
$$\mathrm{Prob}_{|\psi\rangle}(a\in\Delta)=\langle\psi|[a\in\Delta]|\psi\rangle.$$
For unit vectors $|\psi\rangle$, we have
$$|\psi\rangle\in [a\in\Delta]\h\quad\Longleftrightarrow\quad \mathrm{Prob}_{|\psi\rangle}(a\in\Delta)=1.$$
Similarly,
\begin{align*}|\psi\rangle\in [a\in\Delta]^\bot \h&\quad\Longleftrightarrow\quad \mathrm{Prob}_{|\psi\rangle}(a\in\Delta)=0\\
&\quad\Longleftrightarrow\quad |\psi\rangle\in[a\notin\Delta]\h\quad\Longleftrightarrow\quad \mathrm{Prob}_{|\psi\rangle}(a\notin\Delta)=1.
\end{align*}
Note that the negation $p^\bot$ does not mean ``$p$ is not true'', but rather ``$p$ is false!'' (In jargon: we have a choice negation, rather than an exclusion negation.)\footnote{All propositions are thus of the form $\mathrm{Prob}_{|\psi\rangle}(a\in\Delta)=\lambda$, where $\lambda$ is $0$ or $1$. Other probabilities are not included in the logic. One might argue that this makes this quantum logic somewhat crude, as propositions with other values for $\lambda$ would be just as valid from an operational point of view. However, there is still no consensus on how to include propositions of this form in the logic. \cite{redqua}} We see that pure states do not define morphisms $\Pi(\A)\ra{}\{0,1\}$, as one might expect from experience with classical mechanics\footnote{Indeed, there, pure states define a homomorphism of complete Boolean algebras from the logic to the two-valued Boolean algebra.}. Some propositions are neither true nor false for a state $|\psi\rangle$. We will later see that even more is true: even if we replace $\{0,1\}$ by some other Boolean algebra (or a general distributive lattice, for that matter) and lower our expectations, we still cannot assign truth values to the propositions. (Jargon: quantum logic fails to admit Boolean-valued models.) We see that an ignorance interpretation of quantum probabilities is an unacceptable point of view. (See theorem \ref{thm:koclog}.) Surprisingly, quantum logic does satisfy the law of the excluded middle: $p\vee p^\bot=\top$.

One sees the origin of this curiosity not by investigating the meet, which is perfectly well-behaved. For unit vectors $|\psi\rangle$,
\begin{align*}|\psi\rangle\in([a\in\Delta]\wedge[a'\in\Delta '])\h&\quad\Longleftrightarrow\quad \mathrm{Prob}_{|\psi\rangle}(a\in\Delta)=1\;\mathrm{and}\;\mathrm{Prob}_{|\psi\rangle}(a'\in\Delta')=1 \\
&\quad\Longleftrightarrow\quad \mathrm{Prob}_{|\psi\rangle}(a\in\Delta\;\mathrm{and}\; a'\in\Delta')=1,
\end{align*}
but rather by investigating the join:
\begin{align*}|\psi\rangle\in([a\in\Delta]\vee[a'\in\Delta '])\h&\quad\Longleftrightarrow\quad \mathrm{Prob}_{|\psi\rangle}(a\in\Delta\;\mathrm{or}\; a'\in\Delta')=1\\
&\quad\not\Longleftrightarrow\quad \mathrm{Prob}_{|\psi\rangle}(a\in\Delta)=1\;\mathrm{or}\;\mathrm{Prob}_{|\psi\rangle}(a'\in\Delta')=1.
\end{align*}
Even if $\mathrm{Prob}_{|\psi\rangle}(a\in\Delta)$ is neither $0$ nor $1$, i.e. $|\psi\rangle\notin [a\in\Delta]\h$ and $|\psi\rangle\notin [a\in\Delta]^\bot\h$, we still have $|\psi\rangle\in ([a\in\Delta]\vee [a\in\Delta]^\bot)\h=\h$. We see that there is some friction between this quantum concept of truth and our intuition.\\
\\
Now, we should stop to think about the issue of epistemology versus ontology: can we give an operational meaning to these propositions? Since measurements in general destroy a quantum state, this is a subtle issue, most notably for the interpretation of the conjunction and disjunction.

The principal problem is that we cannot perform many measurements successively to verify or falsify our proposition, as measurements change our quantum state. Ideally, therefore, we would like to prepare many copies of the same state and perform parallel experiments. Unfortunately, due to various no-cloning theorems, this is still not possible in principle \cite{niequa}. We can, however, prepare almost identical states\footnote{To an arbitrary precision, in theory.} and perform experiments on them. (See, for instance, \cite{duapro} for the theoretical principle and \cite{cheexp} for a recent experimental realisation of the procedure.) In this way, we can make sense of a probabilistic interpretation of quantum mechanics.

However, the issue of conjunctions and disjunctions remains: which proposition do we verify first? (In quantum mechanics, this will make a big difference\footnote{I am talking about the case of two propositions that involve non-commuting observables. Unless we are in a common eigenstate by coincidence, these observables are not comeasurable.}!) Of course, $\wedge$ and $\vee$ should be symmetric, so there seems to be a problem. The way out of this, proposed by Piron and Jauch, is to interpret $p\wedge p'$ operationally as ``randomly choose to perform an experiment for either $p$ or $p'$''. By carrying out this procedure many times on almost identical states we can hope to verify or falsify this conjunction. We interpret the disjunction in a similar way \cite{balqua}.\\
\\
Finally, one might wonder whether we can again reconstruct $\A$ from $\Pi(\A)$. Put bluntly, does this logic contain all the interesting information about our physics? The answer is yes, if we start out with $\A$ as a von Neumann subalgebra of $\mathcal{B}(\h)$ for some Hilbert space $\h$ and we remember the embedding $\Pi(\A)\subset\mathcal{B}(\h)$\footnote{This is also true in another sense. If we know the truth value of all propositions, we know the support of our state. This means that it uniquely determines our state if and only if it is pure, just as in classical kinematics.}. Then a non-trivial result from functional analysis tells us that $\A$ is isomorphic to the double commutant of $\Pi(\A)$ \cite{takthe}. However, without this extra information, it is not immediately clear how one would reconstruct $\A$. In that sense, the result is somewhat unsatisfactory when one compares it to its classical analogue.\\
\\
Summarising, we have found that quantum logic differs from the classical logic (which we can derive from classical mechanics) in the following respects.
\begin{enumerate}
\item $\vee$ and $\wedge$ do not distribute over each other.
\item Therefore, there is no Heyting implication.
\item The interpretation of $(-)^\bot$ is not what we are used to: for a state $|\psi\rangle$, the truth of $p^\bot$ should be interpreted as ``$p$ is false'', rather than as ``$p$ is not true''.
\item Pure states do not determine the truth of all quantum propositions. Because the notion of truth is too crude, they do not determine an interesting\footnote{Of course, one can choose to call the propositions `true' that are true with probability 1, while one says the rest are `false'.} map from our quantum lattice to some set of truth values (rather than probabilities). Jargon: there is no satisfactory state-proposition pairing.
\item More strongly, there cannot be such a map to a Boolean algebra of truth values if we require it to preserve the logical structure on the Boolean sublogics.
\item The interpretation of $\vee$ is not what we are used to: for a state $|\psi\rangle$, $p\vee p'$ can be true, while neither $p$ nor $p'$ is true. In particular, the law of the excluded middle holds.
\item The operational interpretation of various propositions is somewhat subtle.
\item It is not obvious how one can reconstruct the algebra of observables from this quantum logic.
\end{enumerate}
Of course, this does not mean that quantum logic is not correct! One has to be careful, though, when interpreting quantum propositions.

\section{Topos Quantum Kinematics (Semantics)}
\subsection{A Topos for Quantum Kinematics}
\subsubsection{Bohr's Doctrine of Classical Contexts}
We hope that this discussion of conventional quantum logic has convinced the reader that a new notion of quantum kinematics that is closer to our classical intuition might lead to new insights. Why would one expect that topos theory could be of any use here? Judging by the introductions given in \cite{buttop1}, \cite{doetop1} and \cite{hlstop}, the primary motivation can be found in the need for a suitable object of truth values for the quantum logic. The sentiment seems to be that the all-or-nothing distinction that is made in Birkhoff-von Neumann quantum logic does not do justice to the subtleties of quantum probabilities. The hope, of course, is that this can be realised as the subobject classifier of a suitable topos.

As I will argue, in pursuit of this goal, the topos approaches to quantum kinematics build a framework in a topos that presents the interplay between observables, states and propositions in a way that stays very close to what we know from classical mechanics. Put bluntly, the distinction between classical and quantum kinematics is reduced to the choice of a topos. Therefore, the next question to address is: ``If quantum mechanics does not take place in $\mathrm{Set}$, which topos should we choose?''\\
\\
The papers by Butterfield and Isham and by Landsman et al. motivate this choice from a vague philosophy known as \emph{Bohr's doctrine of classical contexts}. Bohr once phrased this as follows \cite{bohdis}.
\begin{quote}
However far the phenomena transcend the scope of classical
physical explanation, the account of all evidence must be
expressed in classical terms. (...) The argument is simply that
by the word experiment we refer to a situation where we can tell
others what we have done and what we have learned and that,
therefore, the account of the experimental arrangements and of
the results of the observations must be expressed in
unambiguous language with suitable application of the
terminology of classical physics.
\end{quote}
We obtain information about a quantum system by investigating it in various classical contexts. In the mathematical formalism, these classical contexts are represented by commutative subalgebras of our algebra of observables. The idea that these classical contexts should contain a lot of information about a quantum system (at least in the case $\dim \h>2$!) is made rigorous by the following much more recent results\footnote{One should note that the Kochen-Specker theorem actually emphasises a fundamental difference between the classical and the quantum world. Paradoxically enough, it will serve as an important motivation for the topos approach to quantum kinematics that tries to stress analogies with classical kinematics.}.

Apparently, assigning a value to all observables in such a way that it is consistent in each classical context is already too much to ask for:
\begin{theorem}[Kochen-Specker (Observable version),\cite{doekoc}, Theorem 2.4]\label{thm:kocobs} Suppose $\A$ is a von Neumann algebra with no type $I_1$ or $I_2$ summands. Then there does not exist a map $\A\ra{}\C$ that restricts to a unital *-homomorphism on each commutative von Neumann subalgebra.
\end{theorem}
Similarly, as a corollary of the solution to the Mackey-Gleason problem (theorem \ref{thm:macgle}), we see that a state is immediately determined by its values in all classical contexts.
\begin{definition}[Quasi-State]\label{def:quasta} Let $\A$ be a von Neumann algebra. By a quasi-state on $\A$ we mean a map $\A\ra{\omega}\C$ such that its restriction to each commutative von Neumann subalgebra is a state and $\omega(a+ia')=\omega(a)+i\omega(a')$ for all self-adjoint $a,a'\in \A$.\end{definition}
\begin{theorem}[Gleason,\cite{bunmac}]\label{thm:glea} Let $\A$ be a von Neumann algebra without a type $I_2$ summand. Then any quasi-state on $\A$ is actually a state.\end{theorem}
Finally, we have a similar phenomenon at the logical level:
\begin{theorem}[Kochen-Specker (Logical version\footnote{To make this sound more like a statement about logic, the usual, weaker statement of this theorem is that there does not exist a map $\Pi(\A)\ra{}B$ into a non-trivial Boolean algebra $B$ that is a bounded Boolean-algebra homomorphism when restricted to each Boolean subalgebra of $\Pi(\A)$.}), \, easy corollary of results in \cite{doekoc}] \label{thm:koclog} Suppose $\A$ is a von Neumann algebra with no type $I_1$ or $I_2$ summands. Then there does not exist a map $\Pi(\A)\ra{}L$ into a non-trivial bounded distributive lattice $L$ that is a bounded lattice homomorphism when restricted to each Boolean subalgebra of $\Pi(\A)$.
\end{theorem}
\begin{proof}Suppose we do have such a map $\Pi(\A)\ra{q}L$. Then we have a bounded lattice homomorphism $L\ra{ay}\mathrm{Spec}(L)$ that embeds the lattice (as the principal ideals) in its frame of ideals. Since this is a frame, it has a natural structure of a (complete) Heyting algebra. The double-negation nucleus gives a frame homomorphism $\mathrm{Spec}(L)\ra{\neg\neg}\mathrm{Reg}(\mathrm{Spec}(L))=:B$ to the Boolean algebra $B$ of regular elements. Then pick a point $p$ of $B$ (a homomorphism to $\{\bot,\top\}$). We can do this since non-triviality of $L$ implies that $B$ is non-trivial and every non-trivial Boolean algebra has points (the points of its Stone space). Then $p\circ\neg\neg\circ ay\circ q$ defines a two-valued finitely additive probability measure on $\Pi(\A)$. According to Lemma 2.7 in \cite{doekoc} this extends to a unital *-homomorphism $\A\ra{}\C$, which we know cannot exist by theorem \ref{thm:kocobs}.\end{proof}

The motivating idea for topos quantum kinematics will therefore be that we investigate quantum systems only by probing them in these classical contexts. Recall that $\mathrm{Set}$ is the topos we use to describe classical mechanics. This suggests that the `quantum topos' might be related to $\mathrm{Set}$. Its objects should represent `things that we can probe by classical contexts to obtain something in $\mathrm{Set}$'. Phrased this way, the doctrine of classical contexts suggests that our topos should be a topos of certain presheaves\footnote{Indeed, one common way of interpreting the Yoneda lemma is as the statement that presheaves over a category $\mathcal{C}$ are entities modelled on $\mathcal{C}$, entities that we can get to know by probing them with objects of $\mathcal{C}$.} over some category of classical contexts.\\
\\
We have reached the point where the two approaches to topos quantum theory part ways. Butterfield, Isham and Doering model a quantum system by a (non-commutative) von Neumann algebra $\A$ and consider the category of presheaves over the poset $\V(\A)$ of commutative von Neumann subalgebras (ordered by inclusion) as their quantum topos\footnote{It should be noted that these authors have introduced many different approaches in their fairly broad programme of topos-theoretic descriptions of physics. We shall, however, restrict to their account of quantum mechanics that connects best with that of Heunen, Landsman and Spitters.}. Heunen, Landsman, Spitters et al., in contrast, start out with a general C*-algebra $\A$ and construct their quantum topos as the category of presheaves over $\mathcal{C}(\A)^{op}$, where $\mathcal{C}(\A)$ denotes the poset of commutative C*-subalgebras of $\A$\footnote{The reader might object that in this way the topos we use for quantum theory depends on the algebra $\A$ and therefore on the particular system we are considering. On the other hand, $\mathrm{Set}$ is used for \emph{all} classical mechanical systems. This is indeed a strange distinction. However, one might say that `\emph{the} quantum topos' should be the one where we take $\A$ to be the algebra of observables corresponding to all measurable quantities in the universe.}.

\subsubsection{Which subalgebras?}\label{sec:subalg}
We should immediately ask how much information we lose by passing from the non-commutative algebra of observables to the poset of its commutative subalgebras. In the case of von Neumann algebras, a partial answer has been given by Doering:
\begin{theorem}[\cite{doeabe}] Suppose $\A,\A'$ are von Neumann algebras without type $I_2$ summands\footnote{i.e. summands of the form $\mathcal{B}(\C^2)$.}. Then for each order-isomorphism $\V(\A)\ra{f}\V(\A')$ there exists a unique Jordan *-isomorphism\footnote{A $*$-algebra homomorphism that need only preserve the symmetric product, rather than the whole product.} $\A\ra{g}\A'$ such that for all $A\in\V(\A)$, $f(A)=g(A)$.
\end{theorem}
One might expect that we do not preserve the commutator and hence, by Heisenberg's equation, the dynamics. Fortunately, we do not lose more than that\footnote{Of course, this theorem does not imply that we lose the commutator. In fact, we do remember whether two elements commute. It is not obvious, however, how to reconstruct the precise value of the commutator if elements do not commute.}! We can therefore hope to formulate a good theory of quantum kinematics in this topos framework.

If we start out with a general C*-algebra $\A$, however, one might suspect that the von Neumann subalgebras do not contain enough information. This can be seen from the following result of Heunen. Indeed, an arbitrary C*-algebra need not even have any projections.
\begin{theorem}[Essentially theorem 4 in \cite{heucha}\footnote{Heunen proves this for a von Neumann algebra $\A$. However, the same proof holds for a general C*-algebra.}]\label{thm:kalm} Let $\A$ be a C*-algebra. Then $\mathcal{C}(\A)$, together with the inclusion functions rather than just the order relation\footnote{This can also be described as the pair $(\mathcal{C}(\A),\uli{\A})$, where $\uli{\A}$ is the tautological commutative C*-algebra object in $\mathrm{Set}^{\mathcal{C}(\A)}$ (see description of covariant approach).}, contains the projection data $\Pi(\A)$. That is, on the one hand, the Boolean algebras of projections corresponding to elements of $\mathcal{C}(\A)$ assemble into the ordered set of projections (a lattice in the von Neumann case, generally non-distributive), using Kalmbach's
Bundle lemma:
$$\Pi(\A)\cong\varinjlim_{A\in\mathcal{C}(\A)}\Pi(A),\txt{(in $\mathrm{Pos}$)}$$
and, in the von Neumann case, we can retrieve each commutative von Neumann subalgebra of $\A$ (up to isomorphism) as the continuous functions on the Stone spectrum of a Boolean subalgebra of $\Pi(\A)$.
\end{theorem}
Apart from showing that we can recover the Birkhoff-von Neumann quantum logic from our new topos framework of quantum logic, we also see that it is natural to consider the poset of all commutative C*-subalgebras when we are dealing with an arbitrary C*-algebra, since the von Neumann subalgebras simply do not contain enough information. A recent result by Nuiten leads us to believe that the C*-subalgebras actually contain a substantial amount of information about the C*-algebra we started out with. To state this result, we first note that our definition of $\mathcal{C}$ on objects extends to a functor

\[
\begin{tikzcd}[column sep=large, row sep=large]
\mathrm{CStar} \arrow[r, "\mathcal{C}"] & \mathrm{Pos}\\
\A \arrow[d, "h"'] \arrow[r, mapsto] & \mathcal{C}(\A) \arrow[d, "\mathcal{C}(h)"]\\
\A' \arrow[r, mapsto] & \mathcal{C}(\A'),
\end{tikzcd}
\]
where $\mathcal{C}(h)$ takes the direct image of subalgebras under $h$. This is again a C*-algebra by Theorem 4.1.9 in \cite{kadfun0}.
Let us adopt the convention that $0$ is not called a C*-algebra. Then we have the following. (This convention will also help avoid difficulties caused by the fact that $0$ has an empty Gelfand spectrum.)
\begin{theorem}[\cite{nuiboh}] The functor $\mathrm{CStar}\ra{\mathcal{C}}\mathrm{Pos}$ is faithful and reflects all isomorphisms.
\end{theorem}
This shows that the poset of commutative C*-subalgebras should contain enough information to set up a formalism for physics, even if we are dealing with an arbitrary C*-algebra. These three theorems validate the choices of sites\footnote{Of course, one can still argue about Grothendieck topologies.} for the quantum toposes made by the two different approaches to quantum kinematics.\\
\\
In this paper, we shall assume that our quantum observables form a von Neumann algebra, perhaps by constructing a universal enveloping von Neumann algebra, if the reader would like to think of it that way. Therefore, it suffices to work with the poset of commutative von Neumann subalgebras in our description of the contravariant approach as well as in that of the covariant approach. This will enable us to associate an interesting logic to the physics.

\subsection{The Contravariant Approach}
As we will see, the more recent covariant approach to topos quantum theory begins by defining an internal commutative C*-algebra in its `quantum topos', starting from the data of the original non-commutative C*-algebra in $\mathrm{Set}$, and then proceeds along the lines of section \ref{sec:obsstapro}, setting up the whole framework (of observables, states and propositions) for kinematics both in terms of algebra and geometry, using internal versions of the duality theorems of this section. By contrast, the contravariant approach was less theory-driven and more ad hoc, skipping the stepping-stone of algebra and immediately proceeding to the construction of a geometry in the topos. In this section, we will discuss the basic definitions of this formalism.

Recall that, in this approach, we assume that we start out with a von Neumann algebra $\A$ of quantum observables. All of our quantum mechanics will take place in the topos $\mathrm{Set}^{\V(\A)^{op}}$.\\
\\
In the contravariant approach, the proposal is to replace the Hilbert space by the so-called \emph{spectral presheaf} as the fundamental object representing the geometry of the quantum system. To define it, we should note the following:
\begin{theorem}[Maximal ideals in a commutative C*-algebra, \cite{strmat}]\label{thm:maxsp} There is a one-to-one correspondence between maximal ideals in a commutative unital C*-algebra $\A$ and non-zero $*$-homomorphisms $\A\ra{}\C$.\end{theorem}
\begin{proof}[Proof (sketch)] Obviously, every such $*$-morphism has a maximal ideal as its kernel. The converse fact that every maximal ideal comes from a morphism to $\C$ is a direct consequence of the Gelfand-Mazur theorem.\end{proof}

\begin{definition}[Spectral presheaf] The spectral presheaf\footnote{I will adhere to the convention of underlining objects of presheaf categories, which is common in the literature on topos quantum logic.} $\uli{\Sigma}\in\mathrm{Set}^{\V(\A)^{op}}$ is defined as
\begin{enumerate}
\item On objects $A\in\V(\A)$, $\uli{\Sigma}(A):=\mathrm{Max}(A)$ is the Gelfand spectrum of $A$.
\item On morphisms $A\subset A'$, $\uli{\Sigma}(A\subset A'):=(\lambda\mapsto\lambda|_{A})$.
\end{enumerate}
\end{definition}
The idea is that this object represents the geometries associated with all the classical contexts bundled together. The construction in \cite{doetop2} then associates a kind of logical structure with this `geometry'. The intuition behind the construction seems to be that the spectral presheaf is a presheaf of Stone spaces\footnote{The Gelfand spectrum of a commutative von Neumann algebra is even hyperstonean. \cite{takthe}} and therefore its subpresheaves of clopen subsets should represent some sort of logic. This leads to the definition of an internal frame\footnote{It might be better to speak about a complete Heyting algebra here rather than a frame, because the interpretation will be mainly a logical one.} $\uli{\mathcal{P}_{cl}\Sigma}\subset \uli{\mathcal{P}\Sigma}$, such that $$\mathrm{Sub}_{cl}\uli{\Sigma}:=\Gamma(\uli{\mathcal{P}_{cl}\Sigma})\footnote{Here $\Gamma$ denotes the global sections functor.}\subset\mathrm{Sub}
(\uli{\Sigma})$$ consists precisely of the subpresheaves $\uli P$ of $\uli{\Sigma}$ such that $\uli P(A)$ is clopen in $\uli{\Sigma}(A)$ for all $A\in\V(\A)$.

In a recent paper, Wolters \cite{wolcom} introduced another internal frame\footnote{The * is used in the notation since we want to distinguish the corresponding internal locale $\uli{\Sigma^*}$ from the spectral presheaf $\uli{\Sigma}$.}$\uli{\OO\Sigma^*}\subset\uli{\p\Sigma}$ to make the role of the spectral presheaf as representing the quantum geometry in the topos more explicit. One should think of it as a realisation of the spectral presheaf as an internal locale. Although it might not have been entirely clear from \cite{wolcom}, these two subframes of $\uli{\p\Sigma}$ are closely related. The constructions amount to the following.\\
\\
Note that $\uli{y_A}\times\uli{\Sigma}\subset\uli{\Sigma}$ and therefore
\begin{align*}\uli{\p \Sigma}(A)&=\mathrm{Hom}(\uli{y_A}\times\uli\Sigma,\uli\Omega)\\
&\cong\mathrm{Sub}(\uli{y_A}\times\uli{\Sigma})\\
&\subset \mathrm{Sub}(\uli{\Sigma}).
\end{align*}
We define a pair of subfunctors
\begin{definition}[$\uli{\p_{cl}\Sigma}\subset \uli{\OO\Sigma^*} \subset \uli{\p \Sigma}$] We define  $$\uli{\p_{cl}\Sigma}(A):=\left\{ \uli{P}\subset \uli{y_A}\times \uli{\Sigma}\; | \; \textnormal{for all $A'\in\V(\A)$:}\quad\uli{P}(A')\textnormal{ is clopen (as a subset of $\uli{\Sigma}(A')$)}\right\}$$
and
$$\uli{\OO\Sigma^*}(A):=\left\{ \uli{P}\subset \uli{y_A}\times \uli{\Sigma}\; | \; \textnormal{for all $A'\in\V(\A)$:}\quad\uli{P}(A')\textnormal{ is open (as a subset of $\uli{\Sigma}(A')$)}\right\}.$$
\end{definition}
Since $\uli{\OO\Sigma^*}$ will play a very important role in later constructions, it is worth spelling out its definition. To make this more explicit, we define the following bundle of topological spaces.
\begin{definition}[(Contravariant) spectral bundle] Let $\Sigma^*$ be the topological space with underlying set $\{ (A,\lambda) | A\in \V(\A),\; \lambda\in \uli{\Sigma}(A)\, \}$ and opens $U\subset\Sigma^*$ such that, when we write $U_A:=U\cap\uli\Sigma(A)$,
\begin{enumerate}
\item $\forall A\in\V(\A): U_A\in\OO\uli{\Sigma}(A)$;
\item If $\lambda\in U_A$ and $A'\subset A$, then $\lambda|_{A'}\in U_{A'}$.
\end{enumerate}
Let us endow $\V(\A)$ with the anti-Alexandroff topology (consisting of the lower sets). Then it is straightforward to verify that the projection map
$$\Sigma^*\ra{\pi}\V(\A)$$
is continuous. We shall call this map the spectral bundle.
\end{definition}

\begin{theorem}[\cite{wolcom}]\label{thm:intfrm}$\uli{\OO\Sigma^*}$ is an internal frame in $\mathrm{Set}^{\V(\A)^{op}}$. We shall write $\uli{\Sigma^*}$ for it when we want to think of it as an internal locale.
\end{theorem}
\begin{proof}It is almost tautological that
$$\uli{\OO\Sigma^*}(A)=\OO(\pi^{-1}(\downarrow A))=\OO (\Sigma^*|_{\downarrow A}). \qquad (*)$$
Now, we recall that by the comparison lemma (see e.g. \cite{johske2} C2.2.3) we have an equivalence of categories between the presheaves over a poset and the sheaves over that poset endowed with the anti-Alexandroff topology. In particular,

\[
\begin{tikzcd}[column sep=large, row sep=large]
\mathrm{Set}^{\V(\A)^{op}} \arrow[r, "\cong"] & \mathrm{Sh}(\V(\A))\\
\uli{P} \arrow[r, mapsto] & \oli{P}\footnote{Again, this overline is a notational convention in the topos quantum logic literature.},\; \oli{P}(\downarrow A)=\uli P (A)
\end{tikzcd}
\]
which restricts to an equivalence of categories between the categories of internal locales. Moreover, by theorem C1.6.3 in \cite{johske2}, we have an equivalence of categories

\[
\begin{tikzcd}[column sep=large, row sep=large]
\mathrm{Loc}(\mathrm{Sh}(X)) \arrow[r, "\cong"] & \mathrm{Loc}/X \\
\oli L \arrow[r, mapsto] & (\oli L(X)\ra{}X),
\end{tikzcd}
\]
where the total space of the bundle of locales corresponding to an internal locale $\oli L$ is $\oli L(X)$. We apply this with $X=\V(\A)$. In $(*)$, we have already found that $$\oli{\Sigma^*}(\V(\A))=\varinjlim_{A\in\V(\A)} \oli{\Sigma^*}(\downarrow A)=\varinjlim_{A\in\V(\A)} \uli{\Sigma^*}(A)=\Sigma^*\ra{\pi}\V(\A)$$ is a bundle of locales (in $\mathrm{Set}$). This makes $\oli{\Sigma^*}$ and therefore $\uli{\Sigma^*}$ into an internal locale.\end{proof}
Similarly, $\uli{\p_{cl}\Sigma}$ is a complete Heyting algebra. (We think of it this way rather than as a locale.)
\begin{theorem}$\uli{\p_{cl}\Sigma}$ is an internal complete Heyting algebra.\end{theorem}
\begin{proof}Note that, by spectral theory, the lattice of clopen subsets of $\uli\Sigma(A)$ is isomorphic to the frame of projectors. (See also theorem \ref{thm:propemb}.) Since $A$ is a von Neumann algebra, the frame of projectors is complete. Consequently, $\p_{cl}(\uli\Sigma(A))$ is a frame. The argument of theorem \ref{thm:intfrm} shows that $\uli{\p_{cl}\Sigma}$ is an internal frame. Finally, we invoke lemma \ref{lem:models} to conclude that it is an internal Heyting algebra.\end{proof}

One might wonder whether the internal locale $\uli\Sigma^*$ can be constructed from an internal commutative C*-algebra, as in the covariant approach. The answer turns out to be negative in all interesting cases, as was proved by Wolters.
\begin{theorem}[\cite{wolcom}] The internal locale $\uli{\Sigma^*}$ is compact. However, if $\A$ is such that $\V(\A)\neq\{\C\cdot 1\}$, then it fails to be regular. In particular, it is not the internal Gelfand spectrum of some internal commutative C*-algebra.\end{theorem}
\quad\\
\\
We will see that the spectral presheaf and its associated internal locale are the fundamental objects representing the geometry in the contravariant approach. That is, we will define the propositions, observables and states in terms of them.

In particular, we will later see that, using a process called daseinisation, we can construct an embedding of sup-lattices

\[
\begin{tikzcd}[column sep=large]
\Pi(\A) \arrow[r, hook] & \Gamma(\uli{\p_{cl}\Sigma})=\mathrm{Sub}_{cl}(\uli{\Sigma}) \subset \Gamma(\uli{\OO\Sigma^*})=\mathrm{Sub}_{open}(\uli{\Sigma}).
\end{tikzcd}
\]
Note that (in general) the meet cannot be preserved as well, since $\Pi(\A)$ is not distributive as a lattice, while $\mathrm{Sub}_{cl}(\uli{\Sigma})$ is. We gain distributivity at the cost of the law of the excluded middle.

Moreover, it turns out that the construction of daseinisation naturally extends to self-adjoint elements of $\A$. Using this construction, observables take the form of arrows $\uli{\Sigma}\ra{}\uli\R^{\leftrightarrow}$, where $\uli\R^{\leftrightarrow}$ is a kind of real-number object in $\mathrm{Set}^{\V(\A)^{op}}$. We will have an embedding $\A_{sa}\longhookrightarrow \mathrm{Hom}(\uli{\Sigma},\uli\R^{\leftrightarrow})$.

Finally, one can also express quantum states in terms of the spectral presheaf. Given this geometry, one might try to mimic classical mechanics or Hilbert space quantum mechanics, where the pure states are represented by points of the geometry, and hope that the points of this internal locale give a good notion of state of our quantum system. It turns out, as was first noted by Butterfield and Isham, that this analogy cannot possibly hold. Indeed, they proved the following.
\begin{theorem}[Contravariant internal Kochen-Specker,\cite{buttop,doekoc}] Suppose $\A$ is a von Neumann algebra with no type $I_1$ or $I_2$ summands. Then the spectral presheaf $\uli\Sigma$ does not have any global sections.\end{theorem}
\begin{proof}This is just a reformulation of theorem \ref{thm:kocobs}. One only has to note the correspondence between maximal ideals in a commutative unital C*-algebra and non-zero $*$-homomorphisms to $\C$.\end{proof}
\begin{remark}In \cite{wolcom}, it is claimed that, as a consequence of this, the internal locale $\uli\Sigma^*$ has no points\footnote{That is, internal locale morphisms from the terminal internal locale $\Omega$. Using the equivalence $\mathrm{Loc}(\mathrm{Sh}(X))\cong \mathrm{Loc}/X$, one could equivalently say that $\Sigma^*\ra{}\V(\A)$ has no global sections \emph{in the sense of maps of locales}.}. However, I have trouble following the arguments that are given. For instance, it seems to be implicitly assumed that $\Sigma^*$ and $\V(\A)$ are sober. Indeed, Wolters assumes that the internal locale would have a point and tries to derive a contradiction by stating (after definition 2.1) that this would imply that \emph{the bundle of topological spaces} $\Sigma^*\ra{}\V(\A)$ has a section. At the same time, he goes to considerable trouble to obtain a partial result on the sobriety of $\Sigma^*$ later on in the paper (Lemma 2.26). Moreover, it is not hard to see that $\V(\A)$ is sober if and only if it is well-founded as a poset (every non-empty subset has a minimal element), which is certainly not true for all infinite-dimensional $\A$. (As an easy consequence, $\Sigma^*$ is also sober if and only if $\V(\A)$ is well-founded.) If this were the only problem, the proof would at least still hold for finite-dimensional $\A$. However, there seems to be a crucial mistake in the point-set topology that follows (at the top of page 14). (This seems to be a consequence of an attempt to mimic the proof of the corresponding theorem in the covariant approach, where the topology of the spectral bundle is defined `in reverse', because one replaces $\V(\A)$ by $\V(\A)^{op}$.) I have not yet found an alternative proof or counterexample (as these are necessarily rather complicated).\end{remark}

The construction of states in the contravariant approach is unfortunately somewhat more intricate, but still resembles the situation of classical kinematics quite closely if one takes the right point of view: states are represented by certain maps $\mathrm{Sub}_{cl}(\uli\Sigma)\ra{}\mathrm{Hom}_{\mathrm{Pos}}(\V(\A)^{op},[0,1])$\footnote{Originally \cite{doetop2} proposed two different notions of a state. As it turns out, however, both are special cases of the definition of a state we shall be using, which is due to \cite{wolcom}.}. Note that in this quantum logic, therefore, states indeed pair with propositions (elements of $\mathrm{Sub}_{cl}(\uli\Sigma)$) to yield truth values in $\mathrm{Hom}_{\mathrm{Pos}}(\V(\A)^{op},[0,1])$, which in turn give rise to maps $\uli 1\ra{}\uli\Omega$, truth values in a topos-theoretic sense.

\subsubsection{Propositions}
Our goal in this section is to construct an injection $\Pi(\A)\stackrel{\uli\delta}{\longhookrightarrow} \mathrm{Sub}_{cl}(\uli \Sigma)$. As we investigate $\A$ by performing experiments in all the classical contexts $A\in\V(\A)$, the intuition is that $\uli\Sigma$ should represent the idea of the `Gelfand spectrum of $\A$'. (We take the spectrum in each context.) Recall that for a commutative von Neumann algebra the Gelfand spectrum coincides with the Stone spectrum of the Boolean algebra of its self-adjoint idempotents. Therefore, one hopes that it can act as a sort of Stone space for quantum logic.\\
\\
Let us proceed with the construction. First, note that for $p\in A\in\V(\A)$, treating $A$ as an algebra of classical observables, it is reasonable from the point of view of our philosophy to set
$$\uli\delta(p)(A):=\{\lambda\in\uli\Sigma(A)\, | \, \lambda(p)=1\,\}.$$
What should $\uli\delta(p)(A)$ be for an $A$ that does not contain $p$, however? According to our philosophy, we should construct $\uli\delta(p)(A)$ out of the data that someone in the classical context $A$ has available about $p$. We should somehow try to approximate the information captured by $p$ as well as we can by information that is accessible to us from our classical point of view $A$. The key idea will be that we approximate $p$ by some projection in $A$. There are two natural ways of doing this: either we take
$$\delta^o(p)_A:=\bigwedge\{q\in\Pi(A)\, |\, q\geq p\, \},$$
the approximation from above, or we approximate from below:
$$\delta^i(p)_A:=\bigvee\{q\in\Pi(A)\, |\, q\leq p\, \}.$$
We would then like to set
$$\uli\delta(p)(A):=\{\lambda\in\uli\Sigma(A)\, | \, \lambda(\delta(p)_A)=1\,\}.$$
(Note that this agrees with our previous definition.) This should define a subfunctor of $\uli\Sigma$. If we take $\delta=\delta^o$, it does indeed, while this fails for $\delta=\delta^i$. This motivates our choice of $\uli\delta:=\uli\delta^o$\footnote{Since the whole situation is mirrored, we will be using the inner daseinisation in the covariant approach.}. These maps $\delta^o$ and $\uli\delta^o$ are both referred to as \emph{outer daseinisation}.

A physical interpretation of these approximation procedures is the following. The outer daseinisation of a Birkhoff-von Neumann quantum proposition represents, in some sense, its strongest consequence in our classical context, while the inner daseinisation would stand for the weakest antecedent in our classical context that would imply it.

We now have the following. (The proof is not too illuminating.)
\begin{theorem}[\cite{doephy}]\label{thm:propemb}
\[
\begin{tikzcd}[column sep=large, row sep=large]
\Pi(\A) \arrow[r, hook, "\uli\delta^o"] & \mathrm{Sub}_{cl}(\uli\Sigma)\\
p \arrow[r, mapsto] & \{\lambda\in\uli\Sigma(A)\, | \, \lambda(\delta^o(p)_A)=1\,\}
\end{tikzcd}
\]
defines an embedding\footnote{Note that an injective (finite) sup-preserving morphism is automatically an order embedding.} of complete sup-lattices into a complete distributive lattice (preserving $\leq,0,1,\vee$). Note that it does not preserve $\wedge$ in general either, since it is a map from a non-distributive lattice to a distributive one. However, $\mathrm{Sub}_{cl}(\uli\Sigma)$ does have small meets and
$$\uli \delta^o(p\wedge q)\leq \uli\delta^o(p)\wedge \uli\delta^o(q).$$
Moreover, $\uli\delta^o$ does not preserve the negation and $\mathrm{Sub}_{cl}(\uli\Sigma)$ does not in general satisfy the law of the excluded middle.\\
Finally, for each $A\in \V(\A)$, it restricts to an order isomorphism

\[
\begin{tikzcd}[column sep=large]
\Pi(A) \arrow[r, "\cong", "\uli\delta^o(A)"'] & \mathrm{Sub}_{cl}(\uli\Sigma(A)).
\end{tikzcd}
\]
\end{theorem}
Summarising, we have embedded the Birkhoff-von Neumann quantum logic into a complete Heyting algebra, sacrificing the meet and negation and, with them, the law of the excluded middle to gain distributivity.

\subsubsection{Observables}
As we will see, quantum observables inject into $\mathrm{Hom}(\uli\Sigma,\uli \R^\leftrightarrow)$, where $\uli  \R^\leftrightarrow$ is some sort of object in $\mathrm{Set}^{\V(\A)^{op}}$ related to the real numbers. One might expect that quantum propositions, like classical propositions, then arise as the pullback of certain subobjects $\uli\Delta\subset\uli  \R^\leftrightarrow$ along observables $\uli\Sigma\ra{\uli\delta(a)}\uli  \R^\leftrightarrow$. This is indeed the approach advocated in \cite{doetop3} and \cite{doewha}. As far as I know, however, this approach has never been fully worked out.

Instead, the contravariant approach uses the ordinary Birkhoff-von Neumann quantum propositions\footnote{Recall that these, in turn, were constructed from an observable $a\in\A_{sa}$ and a Borel subset $\Delta\subset\R$}, but embedded in $\mathrm{Sub}_{cl}(\uli\Sigma)$, as in theorem \ref{thm:propemb}. The realisation of observables as maps $\uli\Sigma\ra{}\uli \R^\leftrightarrow$ therefore mostly serves to emphasise the role of $\uli\Sigma$ as a quantum analogue of the classical phase space. A second reason for presenting it here is that it originally inspired the covariant approach, where one constructs quantum propositions as in classical mechanics, as inverse images.\\
\\
The first task is to extend the definition of the outer and inner daseinisation from projections to all observables. This is done by extending the order on $\Pi(\A)$ to the so-called \emph{spectral order} $\leq_s$ on $\A_{sa}$. Let $a,a'\in\A_{sa}$ and let $(e_\lambda\in\Pi(\A))_{\lambda\in\R}$ and $(e_\lambda'\in\Pi(\A))_{\lambda\in\R}$ be their respective resolutions. Then we say that $a\leq_s a'$ if and only if $e_\lambda\geq e_\lambda'$ for all $\lambda\in\R$. This order turns the self-adjoint elements into a conditionally complete lattice\footnote{i.e. every set of elements that has an upper bound also has a join and every set of elements that has a lower bound has a meet.} \cite{olssel}. Note that this is not the usual order on $\A_{sa}$. Then, for $a\in\A_{sa}$ and $A\in\V(\A)$, we define the outer and inner daseinisations,
$$\delta^o(a)_A:=\bigwedge\{a'\in A_{sa}\, |\, a'\geq_s a\, \}\txt{and}$$
$$\delta^i(a)_A:=\bigvee\{a'\in A_{sa}\, |\, a'\leq_s a\, \}.$$
Effectively, we replace the projections in the spectral resolution of $a$ by their daseinisations.

Then, $\uli\R^{\leftrightarrow}$ is defined as the subpresheaf\footnote{$\mathrm{Hom}_{\mathrm{Pos}}(\downarrow (A\leq A'),\R)=\mathrm{Hom}_{\mathrm{Pos}}(\downarrow A',\R)\ra{}\mathrm{Hom}_{\mathrm{Pos}}(\downarrow A,\R),\; \mu\mapsto \mu|_{\downarrow A}$ and similarly for $\mathrm{Hom}_{\mathrm{Pos}}((\downarrow -)^{op},\R)$.} of $\mathrm{Hom}_{\mathrm{Pos}}(\downarrow -,\R)\times \mathrm{Hom}_{\mathrm{Pos}}((\downarrow -)^{op},\R)$, where $\uli\R^{\leftrightarrow}$ consists of the $(\mu,\nu)$ such that pointwise $\mu\leq \nu$. The \emph{contravariant daseinisation of observables} is defined to be the map

\[
\begin{tikzcd}[column sep=small, row sep=large]
\A_{sa} \arrow[r, hook, "\uli{\breve\delta}"] & \mathrm{Hom}(\uli\Sigma,\uli\R^{\leftrightarrow}) & \\
a \arrow[r, mapsto] & \Big( \uli\Sigma \xrightarrow{\uli{\breve\delta(a)}} \uli\R^{\leftrightarrow}\Big) &\\
& \lambda\in\uli\Sigma(A) \arrow[r, mapsto] & \big(A'\mapsto \lambda|_{A'}(\delta^i(a)_{A'}),\,A'\mapsto \lambda|_{A'}(\delta^o(a)_{A'})\big)_{A'\in\downarrow A}.
\end{tikzcd}
\]
It is not difficult to see that this is injective. See, for instance, \cite{doewha}.

Trying to mirror classical kinematics, one might hope that $\uli{\breve\delta(a)}$ would be a continuous map or a measurable map in some sense. Then, one would have a reasonable way of constructing propositions as inverse images of open or measurable subobjects of $\uli\R^{\leftrightarrow}$. Indeed, in the covariant approach, one can show that a similar construction helps realise observables as continuous maps or internal locales. In \cite{wolcom}, Wolters shows that if one replaces $\uli\Sigma$ by the internal locale $\uli\Sigma^*$ and if one chooses the most obvious frame of opens for $\uli\R^{\leftrightarrow}$, then the map defined above is not continuous.

\subsubsection{States}
In their original papers, Isham and Doering use two notions of a state that they call `pseudo-states' and `truth objects'. The definition of state we will use here, first introduced in \cite{doequa}, is a more general notion that incorporates both of these notions. It has the added advantage that it can also deal with mixed states and that it resembles the convention in the covariant approach more closely.
\begin{definition}[Measure on an internal lattice $L$] \label{def:intmeas} We call an order-preserving map $L\ra{\mu}[0,1]_l$\footnote{Here $[0,1]_l$ denotes the unit interval in the lower reals (in the topos).}  a (finitely additive) measure if
\begin{enumerate}
\item $\mu(\bot)=0$;
\item $\mu(\top)=1$;
\item $\mu(x)+\mu(y)=\mu(x\wedge y)+\mu(x\vee y)$.
\end{enumerate}
\end{definition}
\begin{definition}[Contravariant state] States in the contravariant approach are represented by internal measures on the spectral presheaf, natural transformations
$$\uli{\p_{cl}\Sigma}\ra{\uli\mu}\uli{[0,1]_l}.$$ They are uniquely determined by $\mu=\Gamma\uli\mu$. They are then characterised as functions
$$\Gamma(\uli{\p_{cl}\Sigma})=\mathrm{Sub}_{cl}(\uli\Sigma)\ra{\mu}
\mathrm{Hom}_{\mathrm{Pos}}(\V(\A)^{op},[0,1])=\Gamma(\uli{[0,1]_l})\footnote{Explicitly $\uli{[0,1]_l}(A)=\mathrm{Hom}_{\mathrm{Pos}}((\downarrow A)^{op},[0,1])$.}$$
such that for every $A\in\V(\A)$ and for every $\uli S_1,\uli S_2\in\mathrm{Sub}_{cl}(\uli\Sigma)$
\begin{enumerate}\item $\mu(\emptyset)(A)=0$;
\item $\mu(\uli\Sigma)(A)=1$;
\item if $\uli S_1\leq\uli S_2$, then $\mu(\uli S_1)(A)\leq\mu(\uli S_2)(A)$;
\item $\mu(\uli S_1)(A)+\mu(\uli S_2)(A)=\mu(\uli S_1 \wedge \uli S_2)(A)+\mu(\uli S_1 \vee \uli S_2)(A)$;
\item $\mu(\uli S)(A)$ only depends on $\uli S_A$.\end{enumerate}

\end{definition}

The first thing we should check is that this definition of a state does indeed have some relation to the conventional notion of a quantum state. We have the following.
\begin{theorem}[\cite{doequa}]\label{thm:contravsta} It is easy to check that a quasi-state $\A\ra{\omega}\C$ defines a measure $\mu_\omega$ on the spectral presheaf by
$$\mu_\omega(\uli S)(A):=\omega\left(\uli\delta^o(A)^{-1}(\uli S(A))\right),$$
where $\uli\delta^o(A)$ is the isomorphism $\Pi(A)\ra{}\mathrm{Sub}_{cl}(\uli\Sigma(A))$ of theorem \ref{thm:propemb}. Less trivially, it can be proved that this defines a bijection between quasi-states on $\A$ and measures on $\uli{\p_{cl}\Sigma}$.

Therefore, by theorem \ref{thm:glea}, if $\A$ has no summands of type $I_2$, the injection of states into measures on the spectral presheaf is also surjective.\end{theorem}
Of course, given that our topos framework was motivated by theorems \ref{thm:kocobs}, \ref{thm:glea} and \ref{thm:koclog}, we could never really hope that it would work out for $\A$ with type $I_2$ summands. This result therefore tells us that measures on the spectral presheaf are precisely the right notion of state in the contravariant approach.

Again, the reader should note that this notion of state is the general notion of a positive normalised functional on $\A$. One can also characterise the states used by physicists, those that come from trace-class operators, in this framework. These are precisely those measures that are in some sense $\sigma$-additive, rather than just finitely additive. The details can be found in \cite{doequa}.

\subsubsection{State-Proposition Pairing}
Recall that one of the issues with Birkhoff-von Neumann quantum logic was that it does not have a satisfactory state-proposition pairing. By Theorem \ref{thm:koclog}, quantum logic does not have any non-trivial Boolean-valued models, not even of the weaker kind where we demand that the logical structure be respected only on Boolean subalgebras of the quantum logic. In Theorem \ref{thm:propemb} we already found that the contravariant topos quantum logic provides some consolation: outer daseinisation defines a Heyting-valued model of this weaker kind.

However, this is not yet exactly what we would hope for. In classical mechanics, each pure state defines a homomorphism from the associated logic to the two-valued Boolean algebra. One might hope that pure states also define homomorphisms to some Heyting algebra of subterminals in this topos framework. This turns out to be too much to hope for. Although we do have a natural map from $\Pi(\A)$ to this Heyting algebra, it is not a homomorphism, not even for pure states. This again indicates a non-deterministic nature of quantum mechanics. In this respect both pure and mixed quantum states behave like mixed states in classical mechanics. The situation is as follows.\\
\\
We note that a quantum state $\omega$ represented by a measure $\mu_\omega$ on the spectral presheaf determines a map

\[
\begin{tikzcd}[column sep=small, row sep=large]
\Pi(\A) \arrow[r, hook, "\uli\delta^o"] & \mathrm{Sub}_{cl}(\uli\Sigma) \arrow[r, "\mu_\omega"] & \mathrm{Hom}_{\mathrm{Pos}}(\V(\A)^{op},[0,1])\\
{\lbrack a\in \Delta\rbrack} \arrow[r, mapsto] & \uli\delta^o([a\in\Delta]) \arrow[r, mapsto] & \Big(A\mapsto \mu_\omega(\uli\delta^o([a\in\Delta])(A))=\omega(\delta^o([a\in\Delta])_A) \Big).
\end{tikzcd}
\]
According to the Born interpretation, we thus assign the probabilities with which the approximation of $[a\in\Delta]$ is true in each classical context. If we do not want to deal with probabilities, one might now apply the map that forgets all the probabilistic information and sends every probability other than $1$ to $0$. This defines a map\footnote{Again, this comes from a map

\[
\begin{tikzcd}[ampersand replacement=\&, column sep=large, row sep=large]
\uli{[0,1]}_l \arrow[r, "{\scriptstyle\uli{\mathrm{Dich}}}"] \& \uli\Omega\\
\mathrm{Hom}_{\mathrm{Pos}}((\downarrow A)^{op},[0,1])=\uli{[0,1]}_l(A) \arrow[r, "{\scriptstyle\uli{\mathrm{Dich}}_A}"] \& \uli\Omega(A)=\{S\subseteq\downarrow A\, | \, S\textnormal{ is a downset}\}\\
\nu \arrow[r, mapsto] \& \{A'\in \downarrow A\, | \, \nu(A')=1 \, \}.
\end{tikzcd}
\]}

\[
\begin{tikzcd}[column sep=large, row sep=large]
\mathrm{Hom}_{\mathrm{Pos}}(\V(\A)^{op},[0,1]) =\Gamma(\uli{[0,1]}_l) \arrow[r, "{\scriptstyle\Gamma\uli{\mathrm{Dich}}}"] & \begin{gathered}\Gamma\uli\Omega\cong\mathrm{Sub}\uli 1\\\cong \{S\subset \V(\A)\, | \, S\textnormal{ is a downset}\, \}\end{gathered}\\
\nu \arrow[r, mapsto] & \{A\in \V(\A)\, | \, \nu(A)=1 \, \}.
\end{tikzcd}
\]
We conclude that we have found a map\footnote{The reader might have been expecting this to be defined as a map $\mathrm{Sub}_{cl}\uli\Sigma\ra{}\mathrm{Sub}\uli 1$, hoping that this would be a lattice (or even Heyting) homomorphism for pure states $\omega$. (The pure states define homomorphisms to the truth values in classical kinematics.) Unfortunately, this does not work out. As one can check, this map indeed preserves $\wedge$, but not $\vee$. This follows since, if it did, we would have found a bounded lattice homomorphism from $\Pi(\A)$ to the distributive lattice $\mathrm{Sub}\uli 1$, which we know cannot exist by theorem \ref{thm:koclog}. (Indeed, $\uli\delta^o$ preserves $\vee$.)}

\[
\begin{tikzcd}[column sep=large, row sep=large]
\Pi(\A) \arrow[r, "\mathrm{truth}_{\omega}"] & \Gamma\uli\Omega=\mathrm{Sub}(\uli 1) \\
p \arrow[r, mapsto] & \Gamma\uli{\mathrm{Dich}}(\mu_\omega(\uli\delta^o(p))).
\end{tikzcd}
\]
\subsubsection{Interpretation of Truth}
Two questions that spring to mind are:
\begin{enumerate}
\item How much information can we infer about $\omega$ from knowing $\mathrm{truth}_\omega$?
\item How should we interpret these truth values?
\end{enumerate}
It is not difficult to see that $\mathrm{truth}_\omega$ contains the same information as the support of $\omega$ (all the usual senses are equivalent: as a linear functional on $\A$, as a measure, or as an operator on $\h$). Therefore, in this respect these `nuanced' truth values are no better than those of the Birkhoff-von Neumann logic. (Recall that there, a state defined a truth value $\Pi(\A)\ra{\mathrm{truth}_\omega}\{0,1\}, p\mapsto \mathrm{Dich}(\omega(p))$.) The difference is therefore to be sought in the truth value of one particular proposition.

By definition\footnote{We will occasionally denote the values this map takes at each context by $0$ and $1$, rather than by $\emptyset$ and $\{*\}$, respectively.}, $\mathrm{truth}_\omega(p)(A)=1$ iff $\omega(\delta^o(p)_A)=1$, i.e. iff the proposition $\delta^o(p)_A$ is true with certainty in state $\omega$ (in the Birkhoff-von Neumann sense). Since $\delta^o(p)_A$ is the strongest logical consequence of $p$ that we can measure in our classical context $A$, we can immediately say that $p$ is true in state $\omega$ if and only if $\mathrm{truth}_\omega(p)(A)=1$ for all $A$. One should therefore think of these propositions with falsification in mind.

Indeed, truth of $\delta^o(p)_A$ guarantees nothing about the truth of $p$. However, assume that the proposition $p$ is not true in state $\omega$ and we want to exhibit this with the restriction that we are only allowed to perform measurements in classical context $A$ (because we want to avoid disturbing our system too much, say). Then our best option is to examine $\delta^o(p)_A$. If that turns out not to be true, neither was $p$. The way to interpret $\mathrm{truth}_\omega(p)$, therefore, is that it has a value $0$ at context $A$ if and only if we can falsify $p$ by measuring observables from $A$. Since the logical operations on the subobject lattice are computed pointwise, the combinations $\uli\delta^o(p)\wedge\uli\delta^o(p')$ and $\uli\delta^o(p)\vee\uli\delta^o(p')$ are formed inside the Boolean logic of the single context $A$. In particular, note that there is no ambiguity about which of the two experiments should be performed first, since the propositions commute.

\subsubsection{Summary}
Summarising, we have set up a framework for quantum kinematics in the topos $\mathrm{Set}^{\V(\A)^{op}}$. This was based on the definition of the spectral presheaf $\uli{\Sigma}$, or the corresponding internal locale $\uli\Sigma^*$. In terms of this `geometry'\footnote{Also, recall that there was no corresponding algebra in the topos, as $\uli\Sigma^*$ failed to be regular and could therefore not be interpreted as the spectrum of an internal C*-algebra.}, we defined generalised sets of observables, states and propositions, in which the corresponding objects of ordinary quantum mechanics were embedded. This went as follows.\\
\\
Observables:

\[
\begin{tikzcd}[column sep=large]
\A_{sa} \arrow[r, hook, "\uli{\breve\delta}"] & \mathrm{Hom}(\uli\Sigma,\uli\R^{\leftrightarrow}),
\end{tikzcd}
\]
(also, recall that this failed to be a map into $C(\uli{\Sigma^*},\uli{\mathbb{IR}})$),
states:

\[
\begin{tikzcd}[column sep=large]
\mathscr{S}(\A) \arrow[r, hook] & \mathscr{S}_{quasi}(\A) \arrow[r, phantom, "\cong"] & \textnormal{``finitely additive measures''}(\uli{\p_{cl}\Sigma}),
\end{tikzcd}
\]
propositions:

\[
\begin{tikzcd}[column sep=large]
\Pi(\A) \arrow[r, hook, "\uli\delta^o"] & \mathrm{Sub}_{cl}\uli\Sigma \arrow[r, hook] & \Gamma \uli{\OO\Sigma^*}\cong \OO\Sigma^*.
\end{tikzcd}
\]
This last map defines a sup-embedding of orders and it restricts to an order-isomorphism between each $\Pi(A)$ and $\mathrm{Sub}_{cl}(\uli\Sigma(A))$, $A\in\V(\A)$. This enlarged `quantum logic' has the following properties.
\begin{enumerate}
\item It is a complete Heyting algebra.
\item The interpretation of the (Heyting) negation is ``$p$ is not true'', rather than ``$p$ is false''.
\item \emph{Each} state determines the truth value of all quantum propositions. This is a consequence of our dichotomy map that transforms probabilities into true/false judgements\footnote{Maybe, `not true' might better describe the situation here than `false'.} in each classical context. The truth values at different classical contexts have the immediate operational interpretation as the possibility of falsification: a $0$ means that falsification is possible.
\item The pure (and mixed) states do not determine lattice homomorphisms $\mathrm{Sub}_{cl}\uli\Sigma\ra{}\mathrm{Sub}\uli 1$.
\item We have a $\mathrm{sup}$-embedding of $\Pi(\A)$ into the complete Heyting algebra $\mathrm{Sub}_{cl}(\uli\Sigma)$, which, at each context $A\in\V(\A)$, restricts to an isomorphism $\Pi(A)\cong\mathrm{Sub}_{cl}(\uli\Sigma(A))$\footnote{Note that this does not contradict theorem \ref{thm:koclog}.}.
\item The operation $\vee$ is computed pointwise in each classical context, rather than by taking a closed linear span. However, the law of the excluded middle fails.
\item The propositions have a clear operational interpretation. The logical operations are defined pointwise, so we are only combining commuting propositions in conjunctions and disjunctions. The truth value of such a combined proposition therefore does not depend on the order in which we test its building blocks.
\item It is not obvious how one can reconstruct the algebra of observables from this quantum logic.

\end{enumerate}
By using outer daseinisation we can form a truth value for each classical context. For a given state, a proposition is not true in a context precisely if it can be falsified by performing measurements from that context.

\subsection{The Covariant Approach}
Conventionally, the covariant approach to topos quantum logic starts with a general C*-algebra $\A$. As discussed in section \ref{sec:subalg}, it is not enough in this case to consider only its commutative von Neumann subalgebras. Instead, one should take into account the poset $\mathcal{C}(\A)$ of all commutative C*-subalgebras. However, recently, particularly in \cite{hlsboh}, the covariant approach has also started to consider more specific kinds of C*-algebras\footnote{In particular, from general to specific, it considered spectral, Rickart, AW* and von Neumann algebras.}. The reason for doing this is that, although one can formulate a satisfactory theory of observables and states when dealing with an arbitrary C*-algebra, the logical side of the story is not yet satisfactory. This is a consequence of the possible lack of projections in a C*-algebra.

Keeping in mind that one can embed any C*-algebra in its universal enveloping von Neumann algebra and that these are in fact the only algebras of operators that mathematical physicists seem to use in practice, I have chosen to restrict to the case of von Neumann algebras, also in my description of the covariant approach. This will also make the comparison with the contravariant approach easier. The reader should bear in mind, however, that most of the results hold in greater generality.

In this approach, we will therefore start out with a von Neumann algebra $\A$ and set up a formalism for quantum kinematics in $\mathrm{Set}^{\V(\A)}$. This time, however, we take a more theory-driven approach. We naturally have an internal commutative C*-algebra in $\mathrm{Set}^{\V(\A)}$. We mirror the constructions we performed in classical kinematics to obtain both a fully algebraic formalism for quantum kinematics in the topos and a geometric counterpart. This internal duality between algebra and geometry will rely on recent constructive versions of Gelfand duality, due to Banaschewski and Mulvey (e.g. \cite{banglo}), and of Riesz-Markov duality, due to Coquand and Spitters (\cite{coqint}) as well as on a less recent constructive version of Stone duality.

\subsubsection{Algebra}
\paragraph{Observables} Given a von Neumann algebra $\A$, we would like to define an object in $\mathrm{Set}^{\V(\A)}$ that will represent this algebra. Note that we have a tautological object\footnote{Recall from theorem \ref{thm:kalm} that the pair $(\V(\A),\uli{\A})$ contains the same information as $\Pi(\A)$ and therefore, by our discussion in paragraph \ref{sec:bn}, one should be able to reconstruct $\A$ from it.} 
\[
\begin{tikzcd}[column sep=large, row sep=large]
\V(\A) \arrow[r, "\uli \A"] & \mathrm{CStar}\\
A\leq A' \arrow[r, mapsto] & A\subset A'.
\end{tikzcd}
\]
Recall the following elementary result from categorical model theory.
\begin{lemma}[e.g. \cite{johske2}, Corollary D1.2.14] \label{lem:models} Let $\mathbb{T}$ be a geometric theory and let $\mathcal{C}$ be a small category. Then

\[
\begin{tikzcd}[column sep=large, row sep=large]
\mathbb{T}-\mathrm{Mod}(\mathrm{Set}^{\mathcal{C}^{op}}) \arrow[r] & \mathbb{T}-\mathrm{Mod}(\mathrm{Set})^{\mathcal{C}^{op}}\\
P \arrow[r, mapsto] & \mathrm{ev}_{-}(P)
\end{tikzcd}
\]
defines an isomorphism of categories.
\end{lemma}
\begin{proof}[Proof] The notation in the lemma for what is essentially the identity functor is supposed to be suggestive. Indeed, note that for each $C\in\mathcal{C}^{op}$ we have a geometric morphism $\mathrm{Set}^{\mathcal{C}^{op}}\ra{\mathrm{ev}_C}\mathrm{Set}$ with inverse image evaluation at $C$ and right adjoint $S\mapsto S^{\mathcal{C}(C,-)}$. This shows that evaluation induces a functor $\mathbb{T}-\mathrm{Mod}(\mathrm{Set}^{\mathcal{C}^{op}}) \ra{} \mathbb{T}-\mathrm{Mod}(\mathrm{Set})^{\mathcal{C}^{op}}$ (which is obviously injective on objects and faithful). Moreover, the set of points $(\mathrm{ev}_C)_{C\in\mathcal{C}}$ is separating. This implies that it is surjective on objects and full.
\end{proof}
Since we have a geometric (even algebraic) theory of commutative rings, $\uli\A$ is an internal commutative ring in $\mathrm{Set}^{\V(\A)}$. We would actually like to say that this is an internal C*-algebra\footnote{We might even want to say that it is an internal von Neumann algebra. However, it is very involved to work out what that would mean. In section \ref{sec:intbool} we shall see that the internal Boolean algebra associated to $\uli\A$ is complete and in section \ref{sec:intstosp} we see that the internal Gelfand spectrum of $\uli\A$ is in fact the internal Stone spectrum of this Boolean algebra. This strongly suggests that $\uli\A$ is at least an internal AW*-algebra.}, so we can apply an internal version of Gelfand duality later. However, since metric completeness is typically a second-order property, a theory of C*-algebras will be second-order, hence we cannot use the above lemma to conclude this. To prove this, one has to work directly with the sheaf semantics. This was done in \cite{hlstop}.
\begin{theorem}[Essentially \cite{hlstop}, theorem 5] Any presheaf of commutative C*-algebras in $\mathrm{cCStar}^{\mathcal{C}^{op}}$, for some category $\mathcal{C}$, is an internal C*-algebra in $\mathrm{Set}^{\mathcal{C}^{op}}$, under the operations inherited from the pointwise ones. In particular, it is an internal vector space over the constant functor $\uli\Delta_\C: c\mapsto \C$ (the Cauchy complex numbers).
\end{theorem}
The process of passing from the non-commutative C*-algebra $\A$ to this internal commutative one is known as \emph{Bohrification} in the literature.
Again, as one might expect, the observables are given by the self-adjoint elements.
\begin{theorem}The object of observables $\uli \A_{sa}$ is given by $\uli\A_{sa}(A)=A_{sa}$.\end{theorem}
\begin{proof}We define $\uli\A_{sa}:=\{ a\in \uli \A \,| \, a^*=a\, \}\subset\uli\A$. This is immediate from the interpretation of sheaf semantics in a presheaf category.\end{proof}

\paragraph{States}
Now that we have given this internal representation of the algebra of observables, we would like to do something similar for the states.
\begin{definition}[Internal State] We define an internal state on an internal C*-algebra $\A$ in a topos to be a $\C\footnote{The Cauchy complexes.}$-linear map $\A\ra{I}\C$, such that $I(1)=1$ and $I(a^*a)\geq 0$ for all generalised elements $a\in \A$.\end{definition}
We have the following correspondence.
\begin{theorem}[\cite{hlstop}]\label{thm:intstate} There is a natural bijection between external quasi-states on $\A$ and internal states on $\uli \A$. Therefore, if $\A$ does not have summands of type $I_2$, external states on $\A$ are in natural bijection with internal states on $\uli \A$.
\end{theorem}
\begin{proof}The first statement is almost tautological. Indeed, any quasi-state defines an internal state by restriction to commutative subalgebras. Conversely, each internal state $\uli\omega$ first defines $\omega(a)=\uli\omega_A(a)$ for self-adjoint $a$, where $A$ is a commutative subalgebra containing $a$. This is well-defined by naturality of $\uli\omega$. We then extend by $\omega(a+ib)=\omega(a)+i\omega(b)$ for self-adjoint $a,b$. Naturality of this bijection is immediate. The second statement is just theorem \ref{thm:glea}.
\end{proof}

\paragraph{Propositions}\label{sec:intbool}
Since we already have an internal structure representing our observables, it is also easy to obtain one for the quantum propositions.
\begin{definition}$\V(\A)\ra{\uli{\Pi(\A)}}\mathrm{Set}$ is the subfunctor of $\uli \A$ that sends a commutative C*-algebra $A$ to its Boolean algebra $\Pi(A)$ of self-adjoint idempotents.
\end{definition}
Note that this is an internal Boolean algebra in $\mathrm{Set}^{\V(\A)}$, by lemma \ref{lem:models}, since the theory of Boolean algebras is geometric (even algebraic). Its operations are induced by those on $\Pi(A)$, for $A\in \V(\A)$. Moreover, combining the fact that each $\Pi(A)$ is a complete Boolean algebra with an argument along the lines of theorem \ref{thm:intfrm}, one shows that $\uli{\Pi(\A)}$ is in fact complete.

\paragraph{State-Proposition Pairing}
Let $\uli \A\ra{\uli\omega} \uli\Delta_\C$ be an internal state. We note that this restricts to a map $\uli{\Pi(\A)}\ra{\uli\omega}\uli\Delta_{[0,1]}$. As one might expect, we can recover our state from this. For comparison with the contravariant approach, we interpret this map $\uli{\Pi(\A)}\ra{\uli\omega}\uli\Delta_{[0,1]}$ slightly differently.
\begin{theorem}[\cite{hlsboh}]\label{thm:intsta} There is a natural bijection between:
\begin{enumerate}
\item External quasi-states on $\A$;
\item Internal states on $\uli\A$;
\item External finitely additive\footnote{\cite{hlsboh} states this result with finitely additive replaced by countably additive and their proof skips over many details. I could see that this would be true if a strong version of Gleason's theorem holds. However, one would certainly have to exclude $I_2$ summands to have such a result at one's disposal.} measures on $\Pi(\A)$.
\item Internal finitely additive measures on $\uli{\Pi(\A)}$.
\end{enumerate}
\end{theorem}
\begin{proof}The equivalence between 1. and 2. has already been established in theorem \ref{thm:intstate}. The equivalence between 1. and 3. is the affirmative answer to the quite non-trivial Mackey-Gleason problem. (Theorem \ref{thm:macgle}.) The equivalence between 3. and 4. is a consequence of the fact that an internal valuation on a Boolean algebra takes values in the Dedekind reals $\uli{[0,1]}\subset\uli{[0,1]}_l$. This assertion is lemma 24 in \cite{hlsboh}. Then note that naturality of an internal measure $\uli\mu$ precisely means that $\uli\mu_A(p)=\uli\mu_{A'}(p)$ if $p\in\Pi(A)\cap\Pi(A')$.
\end{proof}
As in the contravariant approach, we can define a pairing between states and propositions that will yield truth values in $\mathrm{Sub}(\uli 1)\cong \Gamma\uli\Omega$. We simply mirror what we do there. For each (quasi-)state $\omega$, we have a map

\[
\begin{tikzcd}[column sep=large, row sep=large]
\Pi(\A) \arrow[r, "\omega\circ\delta^i"] & \mathrm{Hom}_{\mathrm{Pos}}(\V(\A),[0,1])\\
p \arrow[r, mapsto] & \Big(A\mapsto \omega(\delta^i(p)_A) \Big).
\end{tikzcd}
\]
Again, we introduce a `dichotomy map' that will regard the probabilities that quantum mechanics normally produces (by the Born rule) from an all-or-nothing point of view. We define

\[
\begin{tikzcd}[column sep=large, row sep=large]
\uli{[0,1]}_l \arrow[r, "{\scriptstyle\uli{\mathrm{Dich}}}"] & \uli\Omega\\
\mathrm{Hom}_{\mathrm{Pos}}(\uparrow A,[0,1])=\uli{[0,1]}_l(A) \arrow[r, "{\scriptstyle\mathrm{Dich}_A}"] & \uli\Omega(A)=\{S\subseteq\uparrow A\, | \, S\textnormal{ is an upset}\}\\
\nu \arrow[r, mapsto] & \{A'\in \uparrow A\, | \, \nu(A')=1 \, \}.
\end{tikzcd}
\]
Taking the composition of these maps, we find our truth assignment

\[
\begin{tikzcd}[column sep=8em]
\Pi(\A) \arrow[r, "{\scriptstyle\mathrm{truth}_\omega:=\Gamma\uli{\mathrm{Dich}}\circ\omega\circ\delta^i}"] & \Gamma\uli\Omega\cong \mathrm{Sub}(\uli 1).
\end{tikzcd}
\]

\subsubsection{Geometry}
\paragraph{Observables} Since we have an internal commutative C*-algebra in our topos, one might hope to realise this algebra as an algebra of functions of some `internal space' by using Gelfand-Naimark duality. However, the conventional proof of the theorem relies on the axiom of choice\footnote{For instance, the Stone-Weierstrass theorem, which is invoked in the proof, does so.}. Since our topos is (generally) non-Boolean\footnote{Recall that a presheaf topos is Boolean if and only if its site is a groupoid. Our site is a (generally) non-trivial poset.}, we see that the internal axiom of choice fails.

Not all is lost, however, since Banaschewski and Mulvey (and finally Coquand and Spitters) gave a constructive proof of the theorem in a series of recent papers. Their version of the theorem constructs the Gelfand spectrum not as a compact Hausdorff space, but as a compact completely regular locale. To make sense of this, recall that the categories $\mathrm{CHaus}$ of compact Hausdorff spaces and $\mathrm{KRegLoc}$ of compact completely regular locales are equivalent in the presence of the axiom of choice.

\begin{theorem}[Constructive Gelfand-Naimark,\cite{banglo,coqcon}] For any Grothendieck topos $\mathcal{E}$, we have an equivalence of categories $$\mathrm{cCStar}(\mathcal{E})^{op}\mathop{\leftrightarrows}^{\uli{C}(-,\uli \C)}_{\uli{\mathrm{Max}}}\mathrm{KRegLoc}(\mathcal{E}),$$
where $\uli{\C}$ denotes the internal locale\footnote{In our specific case, its corresponding frame $\uli{\OO(\C)}$ is given by $\uli{\OO(\C)}(A)=\OO(\uparrow A \times \C)$, where $\uparrow A\subset \V(\A)$ is equipped with the subspace topology. \cite{hlstop}} of Dedekind complex numbers, $\uli{C}(-,-)$ denotes the internal continuous Hom (a subfunctor of the internal Hom, defined in terms of the internal language), and $\uli{\mathrm{Max}}$ is the functor that sends an internal commutative C*-algebra to the locale of maximal ideals (defined in terms of the internal language of $\mathcal{E}$)\footnote{The interested reader can find an account of its action on morphisms in \cite{nuiboh}.}. In particular, we can view
$$\uli \A_{sa}\cong \uli C(\uli{\mathrm{Max}}(\uli\A),\uli\R),$$
where $\uli\R$ is the locale\footnote{In our case $\uli{\OO(\R)}(A)=\OO(\uparrow A\times \R)\cong \mathrm{Hom}_{\mathrm{Pos}}(\uparrow A,\OO(\R))$.} of Dedekind real numbers.
\end{theorem}
At the moment, this internal spectrum is a rather abstract object. If it is to be of any use to physicists, we should give a more explicit description of it. For this purpose, we shall view it as a bundle of locales (in $\mathrm{Set}$). For convenience, we shall write $\uli{\Sigma_*}:=\uli{\mathrm{Max}}(\uli\A)$. The reader is encouraged to compare it with the internal locale $\uli{\Sigma^*}$ from the contravariant approach.

In analogy with what we did for the contravariant approach, using the comparison lemma, we will make the identification $\mathrm{Set}^{\V(\A)}\cong\mathrm{Sh}(\V(\A))$, where $\V(\A)$ is here equipped with the Alexandroff topology\footnote{This consists of the upper sets.}. We do this since we have a convenient description of locales internal to localic toposes, namely,
$$\mathrm{Loc}(\mathrm{Sh}(X))\cong \mathrm{Loc}/X.$$
As was computed in \cite{heugel}, these identifications give us the following bundle of locales, representing $\uli\Sigma_*$.

\begin{definition}[(Covariant) spectral bundle] Let $\Sigma_*$ be the topological space with underlying set $\{ (A,\lambda) | A\in \V(\A),\; \lambda\in \mathrm{Max}(A)\, \}$ and opens $U\subset\Sigma_*$ such that, when we write $U_A:=U\cap\mathrm{Max}(A)$,
\begin{enumerate}
\item $\forall A\in\V(\A): U_A\in\OO\mathrm{Max}(A)$;
\item If $A\subset A'$, then $\lambda'\in U_{A'}$ whenever $\lambda'|_{A}\in U_{A}$.
\end{enumerate}
Let us endow $\V(\A)$ with the Alexandroff topology (consisting of the upper sets). Then it is straightforward to verify that the projection map
$$\Sigma_*\ra{\pi}\V(\A)$$
is continuous. We shall call this map the spectral bundle.
\end{definition}

Writing out the identifications we made above, we get the following simple description of the internal Gelfand spectrum, in terms of the spectral bundle.
\begin{corollary}As an object of $\mathrm{Sh}(\V(\A))$, $\oli{\Sigma_*}$ has the frame of opens

\[
\begin{tikzcd}[column sep=large, row sep=large]
U \arrow[r, mapsto, "\oli{\OO\Sigma_*}"] & \OO\Sigma_*|_U\\
U\subset V \arrow[r, mapsto] & (W\mapsto W\cap U).
\end{tikzcd}
\]
As an object of $\mathrm{Set}^{\V(\A)}$, $\uli{\Sigma_*}$ has the frame of opens

\[
\begin{tikzcd}[column sep=large, row sep=large]
A \arrow[r, mapsto, "\uli{\OO\Sigma_*}"] & \OO\Sigma_*|_{\uparrow A}\\
A\leq A' \arrow[r, mapsto] & (W\mapsto W\cap \uparrow A').
\end{tikzcd}
\]
\end{corollary}
Note that by interpreting the internal spectra in both approaches as bundles of locales, we are able to compare them, even though they were internal locales in different toposes to begin with.\\
\\
We may have a reasonable idea of what the internal spectrum looks like now, but the whole operation seems of limited use if we do not know how actual physical observables should relate to it. It is somewhat unfortunate that there seems to be no obvious way of interpreting observables as global sections of $\uli C(\uli{\mathrm{Max}}(\uli \A),\uli\R)$, i.e. as elements of the external continuous homset $C(\uli{\mathrm{Max}}(\uli \A),\uli\R)$. In \cite{hlstop}, the authors were able, however, to find a reasonable injection $\A_{sa}\longhookrightarrow C(\uli{\mathrm{Max}}(\uli \A),\uli{\mathbb{IR}})$, where $\uli{\mathbb{IR}}$ denotes the internal Scott interval domain, another kind of real numbers object. In analogy with the contravariant approach, they dubbed this map `daseinisation' as well.

As a poset, this interval domain $\uli{\mathbb{IR}}$ (in $\mathrm{Set}$) is defined as the set of all non-empty compact intervals of real numbers ordered by reverse inclusion. As with any preorder, this can be endowed with the Scott topology. The closed subsets are the lower sets that are closed under suprema of directed subsets. The collection
$$(p,q)_S:=\{[r,s] \, | \, p<r\leq s<q\},\quad p,q\in\mathbb{Q}, p<q,$$
is a basis for this topology. We have a continuous embedding
$$\R \stackrel{\{\cdot\}}{\longhookrightarrow} \mathbb{IR}$$
$$r\mapsto [r,r].$$

We shall use a variant of the covariant daseinisation map that was introduced by Wolters in \cite{wolcom}. This is obtained from the map in \cite{hlstop} by replacing the (usual) order on $\A_{sa}$ by the spectral order. We have made this choice, on the one hand, because it is compatible with the ordinary Gelfand transform $A\ra{\eta_A}C(\mathrm{Max}(A),\C)$ in the sense that for all $a\in \A_{sa}$, if we write $(a)$ for the (commutative) subalgebra generated by $a$,

\[
\begin{tikzcd}[column sep=large, row sep=large]
\Sigma_*|_{\uparrow (a)} \arrow[rr, "\delta(a)"] \arrow[dr, "\eta_{(a)}"'] & & \mathbb{IR}\\
& \R \arrow[ur, hook, "\{\cdot\}"'] &
\end{tikzcd}
\]
On the other hand, it will turn out that by making this choice, the resulting way of constructing propositions out of observables and intervals coincides with the one described in section \ref{sec:bn}. (See theorem \ref{thm:wolcom}.)

\begin{definition}We define the \emph{covariant daseinisation map} to be the function (which is easily seen to be injective by injectivity of the (inner and outer) daseinisations)

\[
\begin{tikzcd}[column sep=small, row sep=large]
\A_{sa} \arrow[r, hook, "\delta"] & C(\Sigma_*,\mathbb{IR}) & \\
a \arrow[r, mapsto] & \Big(\Sigma_* \xrightarrow{\delta(a)} \mathbb{IR}\Big) & \\
& (A,\lambda) \arrow[r, mapsto] & {[\lambda(\delta^i(a)_A),\lambda(\delta^o(a)_A)]}.
\end{tikzcd}
\]
One can also define the interval domain internal to $\mathrm{Set}^{\V(\A)}$, for instance as an internal locale, by using the internal language. (See \cite{wolcom} for an explicit description.) The specifics do not matter here. The result, however, is
that we also find a map $\A_{sa}\ra{\uli\delta}C(\uli \Sigma_*,\uli{\mathbb{IR}})$.
\end{definition}
We will use this covariant daseinisation map to construct elementary propositions out of the combination of an observable and a subset of the outcomes.

\paragraph{States} As might be expected, there is also a geometric interpretation of states. This precisely reflects the one given by Riesz-Markov duality in classical kinematics. However, to avoid introducing the additional definition of an `internal Borel $\sigma$-algebra associated to an internal locale', we note that finite regular measures on the Borel $\sigma$-algebra can be described by probability valuations on the frame of open sets. This motivates the following definition.
\begin{definition}[Probability valuation on an internal frame $F$] We call an order-preserving map $F\ra{\mu}[0,1]_l$\footnote{Here $[0,1]_l$ denotes the unit interval in the lower reals (in the topos).}  a probability valuation\footnote{cf. definition \ref{def:intmeas}.} if
\begin{enumerate}
\item $\mu(\bot)=0$;
\item $\mu(\top)=1$;
\item $\mu(x)+\mu(y)=\mu(x\wedge y)+\mu(x\vee y)$;
\item $\mu(\bigvee_{i\in I}x_i)=\bigvee_{i\in I}\mu(x_i)$, for any directed set $(x_i)_{i\in I}$.
\end{enumerate}
\end{definition}
As was recently proved by Coquand and Spitters, the Riesz-Markov theorem holds in the topos setting. We formulate a weak version of it\footnote{i.e. one that does not take into account the structure of the space of valuations as a locale or combine it immediately with Gelfand-Naimark.}.
\begin{theorem}[Constructive Riesz-Markov,\cite{coqint}]\label{thm:intrie} For any commutative C*-algebra $\uli\A$ internal to a Grothendieck topos $\mathcal{E}$, there is a natural bijection between internal states on $\uli\A$ and probability valuations on the frame of opens of its internal Gelfand spectrum.
\end{theorem}
This gives yet another characterisation of states in the covariant approach (cf. theorem \ref{thm:intsta}): as internal probability valuations $\uli{\OO\Sigma_*}\ra{\uli\mu}\uli{[0,1]_l}$. Again, these are uniquely determined by their global sections $\OO(\Sigma_*)\cong\Gamma\uli{\OO\Sigma_*}\ra{\mu=\Gamma\uli\mu}\mathrm{Hom}_{\mathrm{Pos}}(\V(\A),[0,1])$. The reader is encouraged to note the similarities with both classical kinematics and the contravariant approach.

However, the analogy with classical mechanics fails in the sense that, here, pure states are not derived from points of the space. Indeed, we have the following reformulation of the Kochen-Specker theorem.
\begin{theorem}[\cite{doekoc}, \cite{hlstop}] Suppose $\A$ is a von Neumann algebra with no type $I_1$ or $I_2$ summands. Then the internal Gelfand spectrum $\uli\Sigma_*$ in $\mathrm{Set}^{\V(\A)}$ does not have any points as an internal locale.
\end{theorem}
\begin{proof}[Proof (sketch)] The idea is that a point\footnote{Here $\uli *$ denotes the terminal internal locale, i.e. the one with corresponding frame the subobject classifier.} $\uli * \ra{\uli\rho}\uli\Sigma_*$ defines a map $\A\ra{}\mathrm{pt}(\uli{\C})\cong\C$ that turns out to restrict to a unital *-homomorphism on each commutative subalgebra, i.e. a map that cannot exist by the Kochen-Specker theorem. (The converse is easy to see.) Indeed, such a point defines a map $\uli C(\uli\Sigma_*,\uli\C)\ra{}\uli C(\uli *,\uli\C)$. Recall that $\uli\A\cong \uli C(\uli\Sigma_*,\uli\C)$, to see that we get a map $\uli\A\ra{}\uli C(\uli *,\uli\C)$. That is, in components, compatible maps $A\ra{}\C$ for all $A\in\V(\A)$. It can be shown that these are unital *-homomorphisms. (This is quite non-trivial.) As we know, these cannot exist by the Kochen-Specker theorem.\end{proof}

We also have a reformulation of the Kochen-Specker theorem in terms of the spectral bundle.
\begin{theorem}[\cite{doekoc},\cite{hlstop}] Suppose $\A$ is a von Neumann algebra with no type $I_1$ or $I_2$ summands. Then the bundle $\Sigma_*\ra{}\V(\A)$ does not admit global sections\footnote{A mathematician might wonder what the added value is of stating the result like this. Such a form of the result is appealing to a physicist since it resembles a very fundamental phenomenon from (Yang-Mills) gauge theory, called Gribov ambiguity. There, the fact that a certain principal bundle does not admit global sections (i.e. is non-trivial) results in the impossibility of making a global choice of gauge (that would fix the value of observables). We have a similar phenomenon here: we cannot globally fix the value of all observables. Regardless of the (mathematical) structural similarity, however, the physical origins of the effects are very different.} (even as maps of locales).
\end{theorem}
\begin{proof}The idea is that under the identification $\mathrm{Loc}(\mathrm{Sh}(X))\cong \mathrm{Loc}/X$ points of an internal locale on the left-hand side correspond to sections of the bundle of locales on the right-hand side. Indeed, a point of an internal locale $L$ in $\mathrm{Sh}(X)$ is a map of internal frames $\OO L\ra{}\Omega$, and under the correspondence this is exactly a section of the associated bundle of locales over $X$. We have already shown that the internal spectrum is an internal locale with no points.\end{proof}
\begin{remark}As far as we know, this result first appeared in \cite{heugel}. However, that proof rests on the unexplained claim that every section of the bundle as a map of locales comes from a section as a map of topological spaces. Since the spaces $\V(\A)$ and $\Sigma_*$ are sober if and only if $\V(\A)$ is co-well-founded (every non-empty subset has a maximal element), this claim is not immediately obvious to the author. Note, however, that it does arise as a corollary of this result.\end{remark}

\paragraph{Propositions}\label{sec:intstosp}
The following result of \cite{hlsboh} gives us hope that we will be able to give a construction of the internal quantum phase space using only the quantum logic, as opposed to the original construction as the internal Gelfand spectrum of the internal commutative C*-algebra of quantum observables.
\begin{theorem}\label{thm:stonspe} The frame of opens $\OO(\mathrm{Max}(A))$ of the Gelfand spectrum of a commutative Rickart C*-algebra\footnote{For our purposes the following is a practical definition: a commutative C*-algebra whose Gelfand spectrum is Stone and whose Boolean algebra of projections is countably complete. \cite{berbae}} $A$ is isomorphic to the frame $\mathrm{Spec}(\Pi(A))$ of ideals of $\Pi(A)$.
\end{theorem}
We would like to go one step further than \cite{hlsboh} and interpret the consequences of this result internally in our topos. To be precise, we will apply internal Stone duality\footnote{Stone duality holds constructively if one is satisfied with the Stone space as a locale. (Recall that one needs the Boolean prime ideal theorem, a weak form of AC, to construct the points of the Stone spectrum.)} to the internal Boolean algebra $\uli{\Pi(\A)}$ to obtain another internal locale, which, fortunately, coincides with our internal Gelfand spectrum. This means that we can recover $\uli \A$ from $\uli{\Pi(\A)}$ as the object of continuous functions from its internal Stone spectrum to $\uli\C$.
\begin{theorem}[Constructive Stone-Duality,\cite{johsto}] For any (elementary) topos $\mathcal{E}$, we have an equivalence of categories
$$\mathrm{Bool}(\mathcal{E})^{op}\mathop{\leftrightarrows}^{\uli{\mathrm{Compl}}}_{\uli{\mathrm{Spec}}}\mathrm{Stone}(\mathcal{E}),$$
where we write $\mathrm{Bool}(\mathcal{E})$ for the category of internal Boolean algebras and homomorphisms and $\mathrm{Stone}(\mathcal{E})$ for the category of internal Stone locales\footnote{i.e. zero-dimensional compact locales} and continuous functions and where $\uli{\mathrm{Compl}}$ sends $\Sigma$ to the Boolean algebra of complemented elements of $\OO\Sigma$ and $\uli{\mathrm{Spec}}$ sends a Boolean algebra $B$ to its internal frame of ideals.\end{theorem}
\begin{proof}Noting that the finite elements of a regular frame (zero-dimensionality implies regularity) are precisely the complemented ones, this is an easy corollary of Corollary II.3.3 in \cite{johsto}.\end{proof}
\begin{corollary}The internal Gelfand spectrum $\uli\Sigma_*$ of $\uli{\A}$ coincides with the internal Stone spectrum\footnote{This suggests that, at least if we are working with a von Neumann algebra, the internal Gelfand spectrum should not be viewed as an analogue of the classical phase space. It is a Stone space and therefore zero-dimensional. (In particular, we should hardly expect to formulate any kind of dynamics on it.) We should rather view it as the quantum analogue of the Gelfand spectrum of the universal enveloping von Neumann algebra of the continuous functions on our classical phase space, i.e. its hyperstonean cover. If we are really looking for a `quantum phase space', we might want to look for a characterisation of $C(X)$ in its universal enveloping von Neumann algebra, find the analogous C*-subalgebra of $\mathcal{B}(\h)$ and study its internal Gelfand spectrum. } $\uli{\mathrm{Spec}}(\uli{\Pi(\A)})$ of $\uli{\Pi(\A)}$. Put differently\footnote{Here, $\uli{\mathrm{Compl}(\Sigma_*)}$ denotes the internal Boolean subalgebra of $\uli{\OO\Sigma_*}$ consisting of the complemented elements.}, $\uli{\Pi(\A)}\cong \uli{\mathrm{Compl}(\Sigma_*)}=\{U\in \uli{\OO\Sigma_*}\, | \, U\vee \neg U=\top\}$.
\end{corollary}
\begin{proof}Under the comparison equivalence $\mathrm{Set}^{\V(\A)}\cong\mathrm{Sh}(\V(\A))$, the internal Gelfand spectrum is represented by the spectral bundle $\Sigma_*\ra{}\V(\A)$, whose frame at a context $A$ is $\OO(\Sigma_*|_{\uparrow A})$ as in the preceding corollary. On the other hand, an element of $\uli{\mathrm{Spec}}(\uli{\Pi(\A)})(A)$ is an ideal of the restricted Boolean-algebra object $\uli{\Pi(\A)}|_{\uparrow A}$, not merely an ideal of the single Boolean algebra $\Pi(A)$. Applying theorem \ref{thm:stonspe} in each context $A'\supset A$ identifies these compatible families of ideals with the compatible families of opens of the fibres $\mathrm{Max}(A')$, which are exactly the opens of $\Sigma_*|_{\uparrow A}$. These identifications are natural in $A$, and hence identify the internal Stone spectrum of $\uli{\Pi(\A)}$ with the internal Gelfand spectrum of $\uli\A$.
\end{proof}
One might wonder whether this object $\uli{\mathrm{Compl}(\Sigma_*)}$ could play an analogous role to that of $\uli{\p_{cl}\Sigma}$ in the contravariant approach. The answer is ``no''. Indeed, note that $\Gamma\uli{\mathrm{Compl}(\Sigma_*)}$ is a Boolean algebra, since $\uli{\mathrm{Compl}(\Sigma_*)}$ is an internal Boolean algebra and the theory of Boolean algebras is algebraic. This means that, by theorem \ref{thm:koclog}, we cannot have a morphism $\Pi(\A)\ra{}\Gamma\uli{\mathrm{Compl}(\Sigma_*)}$ that preserves anywhere near as much structure as the embedding of theorem \ref{thm:propemb} does.

The analogy actually does hold, but one has to choose the correct\footnote{The reader should note that the two options coincide in $\mathrm{Set}$.} object of `clopen subobjects' of $\uli\Sigma_*$. Indeed, $\uli{\p_{cl}\Sigma}$ is not the internal Boolean subalgebra of $\uli{\OO\Sigma^*}$ of complemented elements either. To correctly mirror the situation of the contravariant approach, we define the following subfunctor of $\uli{\OO\Sigma_*}$:
$$\uli{\mathrm{Clopen}(\Sigma_*)}(A):=\left\{ D\in \uli{\OO\Sigma_*}(A)\, | \, \textnormal{for all $A'\in\uparrow A:$ $D\cap\mathrm{Max}(A')$ is closed} \right\}.$$
In particular, $\Gamma\uli{\mathrm{Clopen}(\Sigma_*)}$ will play the role of $\mathrm{Sub}_{cl}\uli\Sigma$:
$$\Gamma\uli{\mathrm{Clopen}(\Sigma_*)}=\left\{D\in \OO\Sigma_*\; | \; \textnormal{for all $A\in\V(\A):$ $D\cap\mathrm{Max}(A)$ is closed} \right\}.$$

Mimicking theorem \ref{thm:propemb}, we define

\[
\begin{tikzcd}[column sep=large, row sep=large]
\Pi(\A) \arrow[r] & \Gamma \uli{\mathrm{Clopen}(\Sigma_*)}\\
p \arrow[r, mapsto] & \{(A,\lambda)\in\Sigma_*\, | \, \lambda(\delta^i(p)_A)=1\}.
\end{tikzcd}
\]
Since $p\in A$ for some $A\in\V(\A)$ and therefore $\delta^i(p)_A=p$ for this $A$, we see that this map is an injection and, literally mirroring the proof of theorem \ref{thm:propemb} in \cite{doephy}, we have
\begin{theorem}\label{thm:propembcov}
\[
\begin{tikzcd}[column sep=large, row sep=large]
\Pi(\A) \arrow[r, hook, "\uli\delta^i"] & \Gamma \uli{\mathrm{Clopen}(\Sigma_*)}\\
p \arrow[r, mapsto] & \{(A,\lambda)\in\Sigma_*\, | \, \lambda(\delta^i(p)_A)=1\,\}
\end{tikzcd}
\]
defines an embedding\footnote{Note that an injective (finite) inf-preserving morphism is automatically an order embedding.} of complete inf-lattices into a complete distributive lattice (preserving $\leq,0,1,\wedge$). Note that it does not in general preserve $\vee$, since it is a map from a non-distributive lattice to a distributive one. However, $\Gamma \uli{\mathrm{Clopen}(\Sigma_*)}$ does have small joins and
$$\uli \delta^i(p\vee q)\geq \uli\delta^i(p)\vee \uli\delta^i(q).$$
Moreover, $\uli\delta^i$ does not preserve the negation and $\Gamma\uli{\mathrm{Clopen}(\Sigma_*)}$ does not in general satisfy the law of the excluded middle.\\
Finally, for each $A\in \V(\A)$, it restricts to an order isomorphism

\[
\begin{tikzcd}[column sep=large]
\Pi(A) \arrow[r, "\cong", "\uli\delta^i(A)"'] & \mathrm{Sub}_{cl}(\mathrm{Max}(A)).
\end{tikzcd}
\]
\end{theorem}
One may wonder whether we can define our states, as in theorem \ref{thm:contravsta} of the contravariant approach, as finitely additive measures on this set $\Gamma\uli{\mathrm{Clopen}(\Sigma_*)}$ of clopen subobjects. This indeed works out, and it even agrees with the restriction of our probability valuations.
\begin{theorem}For each quasi-state $\omega$, one can define a finitely additive measure $\mu_\omega$ on $\Gamma\uli{\mathrm{Clopen}(\Sigma_*)}$ by
$$\mu_\omega(\uli S)(A):=\omega\left(\uli\delta^i(A)^{-1}(\uli S(A))\right),$$
where $\uli\delta^i(A)$ is the isomorphism $\Pi(A)\ra{}\mathrm{Sub}_{cl}(\mathrm{Max}(A))$ of theorem \ref{thm:propembcov}. In fact, we can equivalently define states to be finitely additive measures on $\Gamma\uli{\mathrm{Clopen}(\Sigma_*)}$. This abuse of notation is acceptable because this definition precisely coincides with the restriction of the probability valuation $\mu_\omega$ on $\OO\Sigma_*$ defined by internal Riesz-Markov duality.\end{theorem}
\begin{proof}We leave the proof of the first statement to the reader. The second claim follows by exactly the same proof as one uses in the contravariant approach. This can be found in \cite{doequa}. We examine the last claim.

We write $\mu_\omega$ for the probability valuation defined by the Riesz-Markov theorem and show that, for $\uli S\in \Gamma\uli{\mathrm{Clopen}(\Sigma_*)}$, the formula of the theorem holds. According to Lemma 4.6 in \cite{wolcom}, we have, for $A\in\V(\A)$,
$$\mu_\omega(\uli S)(A)=\sup \{\omega(p)\, | \, p\in \Pi(A), X_p^A\subset \uli S(A)\},$$
where $X_p^A=\{\lambda\in\mathrm{Max}(A)\, | \, \lambda(p)>0\, \}$. Now, writing $p_{\uli S(A)}:=\uli\delta^i(A)^{-1}(\uli S(A))$, note that $\uli S(A)=\{\lambda\in \mathrm{Max}(A)\, | \, \lambda(p_{\uli S(A)})=1\}$. Therefore, for each $p$ over which we take the supremum of $\omega$,
$$\forall \lambda\in\mathrm{Max}(A): \lambda(p)>0 \Rightarrow \lambda(p_{\uli S(A)})=1.$$
For projections in a commutative von Neumann algebra, $X_p^A\subseteq \uli S(A)$ is equivalent to $p\leq p_{\uli S(A)}$. Hence every $p$ in the supremum satisfies $\omega(p)\leq\omega(p_{\uli S(A)})$, and $p=p_{\uli S(A)}$ is itself among the projections over which the supremum is taken. Thus the supremum is $\omega(p_{\uli S(A)})$, as required.

Finally, note that, according to Mackey-Gleason, every such measure also defines a quasi-state by restricting along $\uli\delta^i$.\end{proof}
 \quad \\
\\
Although this parallel with the contravariant approach is useful, the situation still does not resemble classical mechanics, even though that was the purpose of this whole endeavour. The framework of \cite{hlstop} does not start with the Birkhoff-von Neumann quantum propositions from section \ref{sec:bn} and embed them into some topos-related logic. Rather, it builds propositions from the data of observables and subsets of outcomes directly in their topos framework, imitating what we did in the classical kinematics of section \ref{sec:clprop}.

For an $a\in C(\Sigma_*,\mathbb{IR})$ and $\Delta\in\OO\mathbb{IR}$, they define a `proposition'
$$[a\in\Delta]_{cov}:=a^{-1}\Delta\in \OO\Sigma_*.$$
If we are to regard this as a proposition, there should be some relation with our Birkhoff-von Neumann quantum propositions that have a genuine motivation in quantum physics and did not just arise in an attempt to force an analogy with classical physics. Let us write $[a\in\Delta]_{BN}$ for the Birkhoff-von Neumann quantum propositions from now on, to make the distinction between the two notions clear. Then, in \cite{wolcom}, Wolters proved the following surprising result. (The proof is not particularly complicated, but it is a rather long exercise in unravelling all the definitions.)
\begin{theorem}[\cite{wolcom}, essentially theorem 4.9]\label{thm:wolcom} The way propositions are formed in the covariant approach to quantum kinematics is compatible with the way dictated by the Born rule for Birkhoff-von Neumann quantum logic:

\[
\begin{tikzcd}[column sep=large, row sep=large]
\A_{sa}\times\OO(\R) \arrow[r, hook, "\delta\times(-)_S"] \arrow[dd, "{[-\in -]_{BN}}"'] & C(\uli\Sigma_*,\uli{\mathbb{IR}})\times\OO(\mathbb{IR}) \arrow[d, "\lbrack-\in-\rbrack_{cov}"]\\
& \OO\Sigma_* \arrow[d, "\oli{(-)}"]\\
\Pi(\A) \arrow[r, hook, "\uli\delta^i"'] & \Gamma\uli{\mathrm{Clopen}(\Sigma_*)}.
\end{tikzcd}
\]
Here $\OO\Sigma_*\ra{\oli{(-)}} \Gamma\uli{\mathrm{Clopen}(\Sigma_*)}$ denotes fibrewise closure (in each fibre the closure of an open is again open, since the spectrum of a commutative von Neumann algebra is extremally disconnected, by theorem \ref{thm:hypsto}) and $(-)_S$ denotes the restriction of the injection\footnote{This is a reasonable map to consider since it restricts to the map $(r,s)\mapsto(r,s)_S$ between the bases of the topologies. Also $(\{x\})_S=\{\{x\}\}$, so it embodies the idea of the map $\R\ra{\{\cdot\}}\mathbb{IR}$.}

\[
\begin{tikzcd}[column sep=large, row sep=large]
\p\R \arrow[r, hook, "(-)_S"] & \p\mathbb{IR}\\
X \arrow[r, mapsto] & \{[r,s]\, | \, [r,s]\subset X\}.
\end{tikzcd}
\]
\end{theorem}
Although this is a rather surprising mathematical relation between the two notions of proposition, we should verify that they represent the same physical idea. To do this, we study the state-proposition pairing and show that they yield the same truth values when combined with an arbitrary state.

\paragraph{State-Proposition Pairing}
From a geometric point of view, the state-proposition pairing in the covariant approach is very close to the one in the contravariant approach. Recall our dichotomy map $\uli{[0,1]}_l  \ra{\uli{\mathrm{Dich}}}  \uli\Omega$.

Each internal valuation $\uli \mu$ on $\uli{\OO\Sigma_*}$ defines a map $\uli{\OO\Sigma_*}\ra{\uli{\mathrm{Dich}}\circ\uli\mu}\uli\Omega$. Since, in this framework, our quantum propositions are maps $\uli 1 \ra{p}\uli{\OO\Sigma_*}$, we can take their composition to yield a truth value

\[
\begin{tikzcd}[column sep=large]
\uli 1 \arrow[r, "{\scriptstyle\uli{\mathrm{Dich}}\circ\uli\mu\circ p}"] & \uli \Omega
\end{tikzcd}
\]
or equivalently

\[
\begin{tikzcd}[column sep=large]
1 \arrow[r, "{\scriptstyle\Gamma\uli{\mathrm{Dich}}\circ\mu\circ p}"] & \Gamma\uli \Omega=\mathrm{Sub}(\uli 1).
\end{tikzcd}
\]
We say that a proposition $p$ is true at context $A$ for a state $\mu$ if $(\Gamma\uli{\mathrm{Dich}}\circ\mu\circ p)(A)=1$. Using this notion of truth, we have the following result.
\begin{theorem}[\cite{wolcom}, Lemmas 4.7 and 4.8 and Theorem 4.9] Regarding truth values, the two ways $[-\in-]_{BN}$ and $[-\in-]_{cov}$ of building propositions are equivalent. The same holds for the algebraic and geometric ways of pairing states and propositions.\\
\\
Let $\mu$ be a valuation $\OO\Sigma_*\ra{}\mathrm{Hom}_{\mathrm{Pos}}(\V(\A),[0,1])$ (one of our generalised quantum states). By theorems \ref{thm:intrie} and \ref{thm:intsta}, we know that $\mu$ corresponds to some quasi-state $\omega_\mu$ on $\A$. Then, for $A\in\V(\A)$, the following are equivalent.
\begin{enumerate}
\item $\omega_\mu(\delta^i([a\in\Delta]_{BN})_A)=1$;
\item $\mu(\uli{\delta}^i([a\in\Delta]_{BN}))(A)=1$;
\item $\mu(\oli{[\delta(a)\in\Delta_S]_{cov}})(A)=1$;
\item $\mu([\delta(a)\in\Delta_S]_{cov})(A)=1$.
\end{enumerate}
\end{theorem}
\begin{proof}Wolters proves this only for the case where the valuation comes from an ordinary state. However, I see no reason why his proof should not work for the general case.\end{proof}
Put this way, for each (quasi-)state $\omega$ and each proposition, we have a map\footnote{The reader might have been expecting this to be defined as a map $\Gamma \uli{\mathrm{Clopen}(\Sigma_*)}\ra{}\mathrm{Sub}\uli 1$, hoping that this would be a lattice (or even Heyting) homomorphism for pure states $\omega$. (The pure states define homomorphisms to the truth values in classical kinematics.) Unfortunately, this does not work out. As one can check, this map indeed preserves $\wedge$, but not $\vee$. This is a consequence of the fact that pure states do not define a one-point measure. (In the GNS-representation the measure of a state $|\psi\rangle$ is defined by the inner product $p\mapsto (A\mapsto \langle \psi | \delta^i(p)_A|\psi\rangle)$.)}

\[
\begin{tikzcd}[column sep=large]
\Pi(\A) \arrow[r, "\mathrm{truth}_\omega"] & \Gamma\uli \Omega\cong\mathrm{Sub}(\uli 1),
\end{tikzcd}
\]
where $\mathrm{truth}_\omega(p)(A)$ is $\mathrm{Dich}$\footnote{$[0,1]\ra{\mathrm{Dich}}\{0,1\}$, $(x<1)\mapsto 0, 1\mapsto 1$.} applied to the left-hand side of any of the equations in the above enumeration (and we have picked some description $[a\in\Delta]_{BN}=p$).

\subsubsection{Interpretation of Truth}
Let us ask the same questions as in the contravariant approach.
\begin{enumerate}
\item How much information can we infer about $\omega$ from knowing $\mathrm{truth}_\omega$?
\item How should we interpret these truth values?
\end{enumerate}

Again, $\mathrm{truth}_\omega$ contains the same information as the support of $\omega$. Therefore, in this respect these `nuanced' truth values are no better than those of the Birkhoff-von Neumann logic. The difference is to be sought in the truth value of one particular proposition.

By definition, $\mathrm{truth}_\omega(p)(A)=1$ iff $\omega(\delta^i(p)_A)=1$, i.e. iff the proposition $\delta^i(p)_A$ is true with certainty in state $\omega$ (in the Birkhoff-von Neumann sense). Since $\delta^i(p)_A$ is the weakest antecedent, in our classical context $A$, that would imply $p$, we can immediately say that $p$ is true in state $\omega$ if and only if there exists an $A$ such that $\mathrm{truth}_\omega(p)(A)=1$. One should therefore think of these propositions with verification in mind.

Indeed, if $\delta^i(p)_A$ is not true, this tells us nothing about the truth of $p$. However, assume that the proposition $p$ is true in state $\omega$ and we want to exhibit this with the restriction that we are only allowed to perform measurements in classical context $A$ (because we want to avoid disturbing our system too much, say). Then our best option is to examine $\delta^i(p)_A$. If that turns out to be true, so is $p$. The way to interpret $\mathrm{truth}_\omega(p)$, therefore, is that it has a value $1$ at context $A$ if and only if we can verify $p$ by measuring observables from $A$.

Since the logical operations on the subobject lattice are computed pointwise, the combinations $\uli\delta^i(p)\vee\uli\delta^i(p')$ and $\uli\delta^i(p)\wedge\uli\delta^i(p')$ are formed inside the Boolean logic of the single context $A$. In particular, note that there is no ambiguity about which of the two experiments should be performed first, since the propositions commute.

\subsubsection{Summary}
Summarising, we have set up a framework for quantum kinematics in the topos $\mathrm{Set}^{\V(\A)}$. This was based on the definition of the tautological internal C*-algebra $\uli\A$ and its internal Gelfand spectrum $\uli\Sigma_*$. In terms of these objects, we defined algebraic and geometric formalisms internal to the topos that deal with generalised kinds of observables, states and propositions, in which the corresponding objects of ordinary quantum mechanics were embedded. This went as follows.\\
\\
Observables:

\[
\begin{tikzcd}[column sep=large]
\A_{sa} \arrow[r, hook, "\delta"] & C(\uli\Sigma_*,\uli{\mathbb{IR}})\cong C(\Sigma_*,\mathbb{IR}),
\end{tikzcd}
\]
states:

\[
\begin{tikzcd}[column sep=small, row sep=small, cells={nodes={font=\small}}]
\mathscr{S}(\A) \arrow[r, hook] & \mathscr{S}_{quasi}(\A) \arrow[r, phantom, "\cong"] & \begin{aligned}
&\textnormal{``internal states''}(\uli\A) \\
&\cong \textnormal{``finitely additive measures''}(\Pi(\A)) \\
&\cong \textnormal{``finitely additive measures''}(\uli{\Pi(\A)}) \\
&\cong \textnormal{``finitely additive measures''} \\
&\quad (\Gamma\uli{\mathrm{Clopen}(\Sigma_*)}) \\
&\cong \mathrm{``valuations''}(\uli\Sigma_*),
\end{aligned}
\end{tikzcd}
\]
propositions:

\[
\begin{tikzcd}[column sep=large]
\Pi(\A) \arrow[r, hook, "\uli\delta^i"] & \Gamma\uli{\mathrm{Clopen}(\Sigma_*)} \arrow[r, hook] & \Gamma \uli{\OO\Sigma_*}\cong \OO\Sigma_*.
\end{tikzcd}
\]
This last map defines an inf-embedding of orders and it restricts to an order-isomorphism between each $\Pi(A)$ and $\mathrm{Sub}_{cl}(\mathrm{Max}(A))$, $A\in\V(\A)$. This enlarged `quantum logic' has the following properties.
\begin{enumerate}
\item It is a complete Heyting algebra.
\item The interpretation of the (Heyting) negation is ``$p$ is not true'', rather than ``$p$ is false''.
\item \emph{Each} state determines the truth value of all quantum propositions. This is a consequence of our dichotomy map that transforms probabilities into true/false judgements in each classical context. The truth values at different classical contexts have the immediate operational interpretation as the possibility of verification: a $1$ means that verification is possible.
\item For pure (and mixed) states, the truth values fail to define lattice homomorphisms $\Gamma\uli{\mathrm{Clopen}(\Sigma_*)}\ra{}\mathrm{Sub}\uli 1$.
\item We have an $\mathrm{inf}$-embedding of $\Pi(\A)$ into the complete Heyting algebra $\Gamma\uli{\OO\Sigma_*}\cong\OO\Sigma_*$, through $\Gamma\uli{\mathrm{Clopen}(\Sigma_*)}$; at each context $A\in\V(\A)$, it restricts to an isomorphism $\Pi(A)\cong\mathrm{Sub}_{cl}(\mathrm{Max}(A))$.
\item The operation $\vee$ is computed pointwise in each classical context, rather than by taking a closed linear span. However, the law of the excluded middle fails.
\item The propositions have a clear operational interpretation. The logical operations are defined pointwise, so we are only combining commuting propositions in conjunctions and disjunctions. The truth value of such a combined proposition therefore does not depend on the order in which we verify its building blocks.
\item We have a rather satisfactory way of reconstructing the internal algebra $\uli\A$ of observables from the `quantum logic' $\uli{\Pi(\A)}$, namely as the continuous functions on its internal Stone spectrum.
\end{enumerate}

\subsection{Discussion}
Recall that the primary motivation\footnote{I should admit that \cite{hlstop} and \cite{hlsboh} give an entirely different motivation for using toposes: a case is made that we should expect quantum logic to be intuitionistic, since the law of the excluded middle cannot possibly hold, as an ignorance interpretation of quantum probabilities is unacceptable. I am afraid to say that I do not entirely follow the reasoning and it leaves me with the impression that the subtleties in the physical interpretation of the logical operations in the different logics are not being fully appreciated.} for attempting to give a topos-theoretic description of quantum kinematics was to provide the ordinary quantitative probabilistic formalism with a qualitative logical counterpart that was less blunt in its notion of truth than the Birkhoff-von Neumann logic. Have these ambitions been realised? Partly, I would say. Possibly more important is the byproduct\footnote{Here, of course, I am mostly speaking of the covariant approach.} of this search for a suitable quantum logic: a framework for quantum kinematics that is surprisingly similar to that of classical kinematics, except for the fact that it is set up in a different topos from $\mathrm{Set}$.\\
\\
Let us recap what we have done in setting up the framework. Essentially, we have restored the commutativity of the algebra and the distributivity of the logic by giving the objects under consideration a more interesting internal structure. Non-classical aspects of quantum kinematics are hidden in the indexing of objects by the poset of classical contexts. In many cases, we were able to show that the constructions from classical kinematics performed in the topos led to (slightly) generalised notions of quantum states, observables and propositions. In the covariant approach, the topos contained a tautological internal commutative C*-algebra from which we managed to develop the whole framework in a way that was entirely analogous to classical kinematics.

By applying daseinisation maps to approximate propositions in other classical contexts, we managed to embed the Birkhoff-von Neumann logic into a complete Heyting algebra, regaining distributivity at the cost of the law of the excluded middle and either the meet or the join.

We defined a new, more subtle notion of truth on these logics, giving a separate judgement in each classical context, with the object of truth values $\Gamma\uli \Omega\cong \mathrm{Sub}(\uli 1)$. In my opinion, the best interpretation of truth values of propositions in the contravariant and covariant approaches is that they encode, respectively, whether falsification is possible and whether verification is possible in different contexts. I have not seen this interpretation of the truth of quantum propositions in the two approaches in the literature. The treatment of the issue is usually limited to a remark along the following lines. Our new quantum propositions relate as follows to the Birkhoff-von Neumann ones. Contravariant propositions (not truth values) represent truths about a system. Covariant propositions represent statements about the system that one might try to verify (but that are not necessarily true) in each context to establish the truth of the corresponding Birkhoff-von Neumann quantum proposition. In particular, truth of a proposition in the covariant sense in some context immediately implies truth in every context in the contravariant approach\footnote{Indeed, for all $A,A'\in\V(\A)$ and $p\in \Pi(\A)$, $\delta^i(p)_A\leq p \leq \delta^o(p)_{A'}$.}.

Unfortunately, the truth value maps corresponding to pure states did not preserve the logical operations, thereby breaking the parallel with classical kinematics. This, again, shows that pure quantum states are fundamentally different from pure states in classical mechanics. They do not seem to be very susceptible to deterministic descriptions.\\
\\
We conclude with a brief comparison of aspects of the two approaches to the subject. One should note that the intentions of the two approaches were initially very different: while the contravariant approach seemed mainly to be looking for a good notion of quantum logic (within their broader framework of topos approaches to physics), the covariant approach focused on stressing the parallel with classical kinematics.

We can still see that this parallel is much more pronounced in the covariant approach, mostly because there is both an algebraic and a geometric side to the story, while the contravariant approach so far has resisted attempts to associate a geometry with it. However, in recent work, notably \cite{hlsboh} and \cite{wolcom}, the logical side of the covariant approach has been explored, mirroring the constructions from the contravariant approach. As far as I can tell, this has been successful. One could even say that the covariant approach now encompasses the contravariant one, with the caveat that the logic has an interpretation in terms of verification rather than falsification.

Interaction has clearly taken place between the two approaches. Even with that in mind, however, it is quite striking that the two produce such similar results, given their completely different origins. This might show that, even though the constructions used are somewhat uncommon by the standards of physics, they can hardly be complete nonsense.

\appendix

\end{document}